\documentclass[reqno]{amsart}

\usepackage{amsmath}
\usepackage{amsthm}
\usepackage{amssymb}
\usepackage{blkarray}
\usepackage{algorithm}
\usepackage{algorithmic}
\usepackage{tikz}
\usepackage{subcaption}
\usepackage{comment}
\usepackage{cancel}
\usepackage{tablefootnote}
\usepackage{url}

\newtheorem{theorem}{Theorem}
\newtheorem{proposition}{Proposition}

\newtheorem{lemma}{Lemma}
\newtheorem{corollary}{Corollary}
\theoremstyle{definition}
\newtheorem{definition}{Definition}

\renewcommand{\vec}[1]{\mathbf{#1}} 
\newcommand{\vecg}[1]{\boldsymbol{#1}} 
\newcommand{\mat}[1]{\mathbf{#1}} 

\newcommand{\blue}{}
\newcommand{\red}{}
\newcommand{\yellow}{}
\newcommand{\gray}{}
\newcommand{\sect}{Section~}

\newcommand\hcancel[2][black]{\setbox0=\hbox{$#2$}%
\rlap{\raisebox{.45\ht0}{\textcolor{#1}{\rule{\wd0}{1pt}}}}#2}

\begin{document}

\title[Resilience Bounds of Network Clock Synchronization]{Resilience Bounds of Network Clock Synchronization with Fault Correction}

\author{Linshan Jiang, Rui Tan, Arvind Easwaran}
\thanks{This manuscript was accepted by {\em ACM Transactions on Sensor Networks} on June 25th 2020. A preliminary version of this work appeared in {\em The 24th IEEE International Conference on Parallel and Distributed Systems (ICPADS)} held in Singapore, December 2018.}
\thanks{The authors were with School of Computer Science and Engineering, Nanyang Technological University, 50 Nanyang Ave, Singapore 639798. E-mail: \{linshan001, tanrui, arvinde\}@ntu.edu.sg}
\thanks{This research was supported in part by an Nanyang Technological University (NTU) Start-up Grant and in part by the Delta-NTU Corporate Laboratory for Cyber-Physical Systems with funding support from Delta Electronics Inc. and the National Research Foundation (NRF), Singapore under the Corp Lab @ University Scheme. The authors acknowledge Mr. Jothi Prasanna Shanmuga Sundaram for discussions during the early stage of this work and Dr. Yi Li for useful feedback on this work.}

\maketitle

\begin{abstract}
  Naturally occurring disturbances and malicious attacks can lead to faults in synchronizing the clocks of two network nodes. In this paper, we investigate the fundamental resilience bounds of network clock synchronization for a system of $N$ nodes against the peer-to-peer synchronization faults. Our analysis is based on practical synchronization algorithms with time complexity down to $O(N^3)$ that attempt to correct the faults by checking the consistency among the following three types of data: 1) the estimated faults, 2) the estimated clock offsets among the nodes, and 3) the measured clock offsets from the potentially faulty peer-to-peer synchronization sessions. Our analysis gives the following three major results. First, the maximum number of faults that can be corrected by the algorithms has a tight bound of $\lfloor N/2 \rfloor - 1$ when every node pair performs a synchronization session. Second, {\yellow by converting the fault resilience problem to a graph-theoretic edge connectivity problem and applying Menger's theorem,} we develop an algorithm to compute the tight bound when not every node pair performs a synchronization session. Third, the number of synchronization sessions to achieve the capability of correcting any $K$ faults has a lower bound of $\lceil N (2K+1) / 2 \rceil$; we also develop an algorithm to schedule the synchronization sessions to approach the lower bound. The above results provide basic understanding and useful guidelines to the design of resilient clock synchronization systems. For instance, our results suggest that, the 4-node network achieves the highest degree of resilience that is defined as the ratio of the maximum number of correctable faults to the number of synchronization sessions. Therefore, by  organizing a large-scale clock synchronization system into a hierarchy of multiple tiers with each consisting of 4-node synchronization groups, we can achieve satisfactory and understood resilience against faults with reduced synchronization sessions.
\end{abstract}

\section{Introduction}
\label{sec:intro}

For {\blue network systems} such as {\blue wireless sensor networks (WSNs)} {\blue and coordinated robots}, accurate clock synchronization among the distributed nodes is important. Correct timestamps make sense {\blue the sensing} data; synchronized clocks enable punctual coordinated operations among {\blue multiple robotic arms that collaborate on a production line in a manufacturing system}. In contrast, desynchronized clocks will undermine system performance and even lead to physical damages {\blue (e.g., clashing of robotic arms)} and system disruptions in time-critical systems. {\blue However, as the distributed nodes are often deployed in complex physical environments with various naturally occurring disturbances and even malicious attacks, maintaining resilient system-wide clock synchronization can be challenging.}

Network Time Protocol (NTP) \cite{mills1991internet} is the foremost means of clock synchronization that is widely known and adopted. The nodes in a system running NTP are organized into a layered hierarchy, in which a stratum-$n$ node acts as a slave in synchronizing itself with a stratum-$(n-1)$ node, and as a master when providing its clock values to stratum-$(n+1)$ and other stratum-$n$ nodes. Thus, in an NTP system, the global time that is directly accessed by the stratum-1 nodes is disseminated to all the nodes stratum by stratum. The intact dissemination highly depends on the successful peer-to-peer (p2p) clock synchronization sessions. {\gray A p2p synchronization session can be realized by the direct communication between the two nodes or by the multi-hop communications via several relay nodes.} The p2p synchronization session estimates the offset between the clocks of two synchronizing nodes (referred to as {\em clock offset} in this paper) using a round-trip timing approach. Specifically, with the one-way packet delivery time estimated {\yellow as} half of the measured round-trip time, the clock offset can be computed based on the two nodes' respective clock values when the packet leaves one node and arrives at the other. With the estimated clock offset, a slave node can reset its clock value or calibrate its clock advance speed to achieve synchronization with the master node. This round-trip timing approach is also the basis of the Precision Time Protocol (PTP) \cite{4579760} that is adopted in industrial Ethernets for higher synchronization accuracy. To reduce communication overhead, the clock synchronization approaches developed for WSNs (e.g., RBS \cite{elson2002fine}, TPSN \cite{ganeriwal2003timing}, and FTSP \cite{maroti2004flooding}) often assume near-zero wireless signal propagation times by accessing radio chips' hardware interrupts and therefore can synchronize two nodes with a one-way communication.

However, the {\yellow integrity} of the p2p synchronization sessions can be compromised. A basis of the round-trip timing approach is that the communication link between the two synchronizing nodes is symmetric, which may not hold faithfully in practice, however. For instance, in a packet-switched network, the uplink to and downlink from the master node may take different routes with distinct end-to-end delays. Over a wireless link, the media access control (MAC) may introduce uncertain delays of up to hundreds of milliseconds in transmitting a packet \cite{maroti2004flooding}. However, not all systems can perform the packet timestamping in the MAC layer to exclude this uncertainty from the clock offset estimation. Moreover, as discussed in RFC 7384 \cite{rfc7384}, the attackers may introduce controlled delays to the deliveries of the packets and breach the symmetric link assumption. As shown in \cite{rabadi2017taming}, this {\em packet delay attack} can be implemented in a wired network via a compromised network router. Moreover, this threat cannot be solved by conventional security measures such as crytographic authentication and encryption \cite{rfc7384,mizrahi2012game,ullmann2009delay}. The violations of the symmetric link assumption caused by any of the above reasons will lead to errors in estimating the clock offset and faulty p2p clock synchronization sessions. The one-way synchronization approaches adopted for WSNs are also susceptible to the delay. In \cite{gu2019attack}, the packet delay attack against a low-power wide-area network is implemented by a combination of malicious packet collision and delayed replay.

In this paper,
we {\blue investigate} the {\blue fundamental} resilience {\blue bounds} of {\em network clock synchronization} (NCS) for a system of $N$ nodes against the p2p synchronization faults.
{\blue Specifically, we study the problem of deriving the maximum number of p2p synchronization faults that the network can correct to maintain the clock synchronization among all the nodes, as well as the dual problem of deriving the minimum number of p2p synchronization sessions to ensure the network's ability to correct a specified number of p2p synchronization faults. Our analysis is based on an NCS algorithm that attempts to correct the p2p synchronization faults. The algorithm is as follows. Consider an {\em NCS graph} $G=(V, E)$, where $V$ and $E$ respectively denote the set of the nodes in the network and the set of the p2p synchronization sessions each performed between two nodes.}  We use $|E|$ to denote the cardinality of the set $E$.
{\yellow The $k$th} step of the algorithm assumes that $k$ out of totally $|E|$ p2p synchronization sessions are faulty, exhaustively tests all possible $|E| \choose k$ combinations of the {\blue assumed} faulty p2p synchronization sessions {\blue among the $|E|$ sessions}, and yields a solution once the estimated clock offsets and the estimated p2p synchronization faults agree with all the p2p clock offset measurements. Starting from $k=0$, the algorithm increases $k$ by one in each step and terminates once a solution is found. {\blue This algorithm is practical} in that it does not require any run-time knowledge about the p2p synchronization faults, including the {\blue actual} number of the faults and {\yellow which synchronization sessions out of totally $|E|$ sessions are actually faulty.} {\gray The implementation of NCS needs a central node that can execute the compute-intensive NCS algorithm and can communicate reliably with all the nodes in the network to collect the results of p2p synchronization sessions. In a sensor network, the gateway with sufficient compute resources can serve as the NCS central node.}

{\yellow However, analyzing the resilience bounds of the NCS algorithm is challenging. First, the approach of analyzing all possible cases of the actual faults and assumed faults incurs prohibitive overhead. Specifically, in the $k$th step of the algorithm, we need to analyze a total of  $|E| \choose k$ possible distributions of the assumed faulty p2p synchronization sessions among the $|E|$ synchronization sessions. Thus, this approach becomes infeasible when $|E|$ is large. Second, for larger networks with more nodes, it becomes difficult to enumerate all possible isomorphic NCS graphs with a certain number of edges. Thus, enumerating all possible cases to analyze the resilience bounds is not a promising approach.
	
	In this paper, to analyze the resilience bounds, we introduce fault-free NCS subgraphs and convert the NCS resilience problem to a graph-theoretic problem. Assisted with the existing results in graph theory, we obtain the following main} {\blue   analytic results  for the resilient NCS problem:
\begin{itemize}
\item For a complete NCS graph in which every node pair performs a p2p synchronization session, the tight bound of the maximum number of p2p synchronization faults that can be corrected by the NCS algorithm is $\left\lfloor \frac{N}{2} \right\rfloor - 1$. In other words, the NCS algorithm can synchronize all nodes when the number of p2p synchronization faults is no greater than the tight bound; otherwise, some nodes in the network will be desynchronized due to the faults.
\item For any NCS graph that may be incomplete, {\yellow we convert the fault resilience problem to a graph-theoretic edge connectivity problem. From our analysis based on Menger's theorem \cite{bohme2001menger},} we develop an algorithm to compute the tight bound of the maximum number of p2p synchronization faults that can be corrected by the NCS algorithm. {\yellow Moreover, we develop a new NCS algorithm with a time complexity of $O(N^3)$ that achieves the same fault correction capability as the original NCS algorithm that has a time complexity of $O\left( \sqrt{2}^N \right)$}.
\item We study the {\em minimum} NCS graph that uses the least edges to provide resilience against a specified number of p2p synchronization faults (denoted by $K$). We develop an algorithm to compute the minimum NCS graphs. We {\yellow prove that} $\left\lceil \frac{N \cdot (2K+1)}{2} \right\rceil$ {\yellow is a lower bound of} the number of edges of any NCS graph that is resilient against $K$ p2p synchronization faults. {\yellow  Numeric results suggest that the lower bound is tight.} The lower bound can be used to understand the order of magnitude of the number of edges in any minimum NCS graph.
\end{itemize}
}

{\blue
  The analytic results in this paper provide important understanding and useful guidelines to the design of clock synchronization systems that are resilient to p2p clock synchronization faults. They are useful to time-critical systems such as industrial wireline and wireless Ethernets. Particularly, in addition to the analysis, this paper discusses the design of a clock synchronization architecture that strikes a good trade-off between the p2p synchronization communication overhead and resilience to p2p faults. Specifically, we use the {\em degree of resilience} (DoR) as the resilience metric, which is defined as the ratio between the number of faults that can be corrected and the number of edges in an NCS graph. Based on our analytic results, we show that a 4-node network with complete NCS graph achieves the highest DoR of $1/6$. From this observation, we propose a tiered clock synchronization architecture for large-scale networks, in which the nodes in a network are grouped into 4-node synchronization groups that are organized into multiple tiers. This architecture provides reduced but well understood resilience against p2p clock synchronization faults (i.e., every 4-node synchronization group can have one fault), compared with the original large-scale network with a minimum NCS graph. The number of p2p synchronization sessions in the proposed architecture is much reduced, compared with that of the minimum NCS graph.
}

The remainder of this paper is organized as follows. \sect\ref{sec:related} reviews related work. \sect\ref{sec:problem} states the problem. {\blue \sect\ref{sec:analysis} analyzes the resilience of several small-scale networks to illustrate the key challenges in the resilience analysis. \sect\ref{sec:fully} derives the tight resilience bound of complete NCS graphs. \sect\ref{sec:lower-bounds} develops the algorithms to compute the resilience bound of any graph. {\yellow Section~\ref{sec:newal} develops} a fast NCS algorithm {\yellow with a cubic time complexity that achieves the same fault correction capability as the original NCS algorithm}. \sect\ref{sec:another} studies minimum NCS graphs. \sect\ref{sec:implication} discusses the implications of our results and a tiered clock synchronization architecture for fault resilience.} \sect\ref{sec:conclude} concludes the paper.

\section{Related Work}
\label{sec:related}

In this section, we review the implementations of clock synchronization for different network systems and the existing studies on the fault tolerance of NCS.

\subsection{Implementations of Clock Synchronization}

Highly stable time sources are often ill-suited for {\blue network systems}. Despite an initial study of using chip-scale atomic clock (CSAC) on {\blue WSN} platforms \cite{dongare2017pulsar}, CSAC is still too expensive (\$1,500 per unit \cite{dongare2017pulsar}) for wide adoption. {\blue Thus, how to synchronize the nodes in different kinds of network systems has received extensive research.}

{\blue On the Internet, NTP \cite{mills1991internet} has been widely used to synchronize computer hosts. It is universal because it imposes few requirements, i.e., it only requires the host to timestamp the transmission and reception of the synchronization packets in the operating system (OS). Compared with NTP, PTP \cite{4579760} additionally requires the network interfaces of the synchronizing hosts and all the switches on the network paths among the hosts to have hardware-level timestamping capability. As such, PTP can exclude the uncertain OS and packet switching delays from the packet delivery time measurements, largely improving the accuracy in estimating the clock offsets. However, malfunctioned hardware-level timestamping will lead to p2p synchronization faults. The reference implementations of NTP and PTP have various mechanisms to improve their robustness against p2p synchronization faults. For instance, in NTP, when the round-trip time exceeds one second, the current p2p synchronization session is considered faulty and will not be used to calibrate the host's clock. Moreover, a slave node will average the clock offset estimation results obtained with multiple master nodes to guide its clock calibration. Despite these heuristics in the protocol implementation for fault resilience, an analytic understanding regarding the network's resilience against the p2p synchronization faults is still lacking.}

{\blue Over the past decade, WSN clock synchronization has been widely studied.} {\blue There are accurate global time broadcasts from the Global Positioning System (GPS) and timekeeping radio stations (e.g., WWVB in U.S.).} However, GPS and radio receivers have various limitations such as high power consumption, poor signal reception in indoor environments (e.g., 47\% good time for WWVB \cite{chen2011ultra}), and susceptibility to wireless spoofing attacks \cite{nighswander2012gps}. Thus, GPS and radio receivers are often employed on a limited number of time masters to provide global time to the slave nodes via some clock synchronization protocol. The resilience of the clock synchronization between the master and the slaves is the focus of this paper.

{\blue Early studies have designed clock synchronization protocols based on message passing, such as RBS \cite{elson2002fine}, TPSN \cite{ganeriwal2003timing}, and FTSP \cite{maroti2004flooding}.} {\blue Recent studies exploit various} external periodic signals {\gray for clock synchronization \cite{viswanathan2016exploiting,rabadi2017taming,yan2017}, time fingerprinting \cite{lukac2009recovering,gupchup2009sundial,viswanathan2016exploiting,li2017natural,maroti2004flooding}, }and clock calibration \cite{rowe2009low,li2012flight,hao2011,li2011exploiting}. Time fingerprinting approaches focus on studying the global time information embedded in the sensing data such as microseisms \cite{lukac2009recovering}, sunlight \cite{gupchup2009sundial}, {\blue powerline voltage \cite{viswanathan2016exploiting}}, and electromagnetic radiation \cite{li2017natural}. They can be a basis for clock synchronization. {\blue For instance, the work \cite{viswanathan2016exploiting} achieves microseconds clock synchronization accuracy by using the time fingerprints found in the electric voltages of a building's power network.}
{\gray Different from clock synchronization that ensures the clocks to have the same value, {\em clock calibration} ensures different clocks to advance at the same speed. The approaches presented in \cite{rowe2009low,li2012flight,hao2011,li2011exploiting} exploit powerline electromagnetic radiation, fluorescent lamp flickering, Wi-Fi beacons, and FM Radio Data System broadcasts to calibrate the clocks of WSN nodes. Clock calibration may not need any inter-node communications, whereas clock synchronization must need communicating one or more timestamps between two synchronizing nodes.}  {\blue However, all the above studies \cite{elson2002fine,ganeriwal2003timing,maroti2004flooding,viswanathan2016exploiting,rabadi2017taming,yan2017,lukac2009recovering,gupchup2009sundial,li2017natural,rowe2009low,li2012flight,hao2011,li2011exploiting} focus on devising clock synchronization/calibration approaches. They fall short of analyzing the resilience of the system against potential clock synchronization faults.}

\subsection{Fault Tolerance of NCS}

The fault tolerance of NCS against Byzantine clock faults has been studied \cite{dolev1984possibility,Lamport1985}. A Byzantine faulty clock gives an arbitrary clock value whenever being read. It has been proved that, to guarantee the synchronization of non-faulty clocks in the presence of $m$ faulty clocks, a total of at least $(3m+1)$ clocks are needed. Different from the Byzantine faulty clock model, we consider faulty p2p synchronization sessions between the clocks. The conversion of our problem to the Byzantine clock synchronization problem by considering either node involving a faulty p2p synchronization session as a faulty clock is {\em invalid}, because this faulty clock after the conversion is not a Byzantine faulty clock, unless all p2p synchronization sessions involving this clock are faulty. As our problem does not have this assumption, the the analysis in \cite{dolev1984possibility,Lamport1985} and the resulted fault tolerance bound are not applicable to our problem. {\blue Moreover, different from the fault-tolerant systems in \cite{dolev1984possibility,Lamport1985} that do not try to correct the faults, our resilient NCS system tries to correct the p2p synchronization faults.}

{\blue Our prior work \cite{tan2018resilience} presented the formulation of the resilience of NCS against p2p synchronization faults. It developed an algorithm to compute a lower bound and derived a closed-form upper bound of the maximum number of faults that can be corrected for any complete NCS graph. In this paper, we derive the closed-form tight bound of resilience for any complete NCS graph, which represents {\yellow a substantial} improvement to the results in \cite{tan2018resilience}. Moreover, this paper studies three new problems: (1) the resilience bounds of NCS graphs that can be incomplete (\sect\ref{sec:lower-bounds}) {\yellow , (2) fast NCS algorithm with polynomial complexity to achieve the same fault correction capability as the original NCS algorithm (Section~\ref{sec:newal}), and (3)} the minimum NCS graphs providing a specified level of resilience (\sect\ref{sec:another}).
}

\section{Problem Statement}
\label{sec:problem}

This section presents the system model (\sect\ref{subsec:system-model}) and states the research problem (\sect\ref{subsec:problem}). In \sect\ref{subsec:background}, we discuss several abstractions in our system model and their relations with real NCS systems.

\subsection{System Model and NCS}
\label{subsec:system-model}

To improve the robustness of clock synchronization against p2p synchronization faults, this section proposes an approach to cross-check the p2p synchronization results among multiple nodes and correct the faults if present.

Let $V$ denote the set of $N$ nodes in a network, i.e., $V = \{n_0, n_1, \ldots, n_{N-1}\}$, where $n_i$ represents the $i$th node. Let $\delta_{ij}$ denote the clock offset between the nodes $n_i$ and $n_j$, which is unknown and to be estimated. Specifically, $\delta_{ij} = c_i(t) - c_j(t)$, where $c_i(t)$ and $c_j(t)$ are the clock values of $n_i$ and $n_j$ at any given Newtonian time instant $t$, respectively. We assume that $\delta_{ij}$ is time-invariant. In \sect\ref{subsec:background}, we will discuss the validity of this assumption in real systems. By designating $n_0$ as the {\em reference node}, we have {\red the relationship }$\delta_{ij} = \delta_{i0} - \delta_{j0}$, {\red which will be used for analysis in the rest of this paper.}

Denote by $n_i \leftrightarrow n_j$ the p2p synchronization session between $n_i$ and $n_j$. Denote by $\widetilde{\delta}_{ij}$ the measured clock offset via $n_i \leftrightarrow n_j$. If the synchronization session $n_i \leftrightarrow n_j$ is successful (i.e., non-faulty), $\widetilde{\delta}_{ij} = \delta_{ij}$; if the synchronization session is faulty, $\widetilde{\delta}_{ij} = \delta_{ij} + e_{ij}$, where $e_{ij}$ is the p2p synchronization fault {\blue which is a non-zero and finite real number}. {\blue Let $E$ denote the set of all p2p synchronization sessions performed in a {\em synchronization round}. In our NCS approach, the result of at most one p2p synchronization session performed between any pair of nodes is used in one synchronization round. Note that the analysis of this paper is agnostic to the technique used for each synchronization session. For instance, a round-trip timing process can be used to obtain $\widetilde{\delta}_{ij}$ between $n_i$ and $n_j$. Moreover, since at most one synchronization session between $n_i$ and $n_j$ is used in one synchronization round, the edge $n_i \leftrightarrow n_j$ can be modeled undirected. }

{\blue For a synchronization round, the undirected graph $G=(V, E)$ is called the NCS graph. In a complete NCS graph, every node pair performs a p2p synchronization session, resulting in $|E|={N \choose 2} = \frac{N(N-1)}{2}$.}
{\blue For any NCS graph $G=(V,E)$ that may be incomplete, the NCS is performed as follows.}
All the clock offset measurements are transmitted to a central node, which runs {\blue the NCS algorithm that is shown in Algorithm~\ref{alg:error-correction}. {\gray The central node can undertake compute-intensive NCS algorithm and can communicate reliably with all the nodes. It can be an external entity (e.g., a cloud service) or any connected node in the network.} For the latter case, the central node may not be the reference node; various strategies can be used to select the central node. For example, in a battery-powered network that concerns about the nodes' energy consumption, a node with the most remaining battery energy can perform the role of central node to receive the clock offset measurements and {\yellow run} the NCS algorithm. We assume that the p2p clock synchronization sessions are separate from the transmissions of the clock offset measurements to the central node. In certain cases, the transmission of the measured clock offset can be avoided. For instance, if the round-trip timing approach is used and the central node initiates the round-trip timing, the central node obtains $\widetilde{\delta}_{ij}$ on the completion of the round-trip timing and requires no separate transmission of $\widetilde{\delta}_{ij}$.} Note that, when every node performs a p2p synchronization session with the reference node, the NCS graph $G$ will have a star topology centered at the reference node. Algorithm~\ref{alg:error-correction} and all analytic results in this paper are also applicable to this star NCS graph.

\renewcommand{\algorithmicrequire}{\textbf{Given:}}
\renewcommand{\algorithmicensure}{\textbf{Output:}}
\renewcommand{\algorithmiccomment}[1]{// #1}

\begin{algorithm}[t]
\caption{{\blue NCS algorithm with fault correction}.}
\label{alg:error-correction}
\begin{algorithmic}[1]
\REQUIRE $\{ \widetilde{\delta}_{ij} | \forall n_i \leftrightarrow n_j \in E \}$
\ENSURE $\{ \hat{\delta}_{j0}| \forall j \in [1, N-1] \}$ and $\{\hat{e}_{ij} | \forall n_i \leftrightarrow n_j \in E \}$

\STATE $k \leftarrow 0$
\WHILE{$k \le |E|$}
\FOR{each distribution of the $k$ estimated p2p synchronization faults among the $|E|$ p2p synchronization sessions}
\label{line:foreach}
\IF{the corresponding variant of Eq.~(\ref{eq:error-correction}) with the $k$ estimated p2p synchronization faults has a solution}
\label{line:solve}
\STATE return $\{ \hat{\delta}_{j0}| \forall j \in [1, N-1] \}$ and $\{\hat{e}_{ij} | \forall n_i \leftrightarrow n_j \in E \}$
\ENDIF
\ENDFOR
\STATE $k \leftarrow k + 1$
\ENDWHILE
\end{algorithmic}
\end{algorithm}

{\blue We now explain Algorithm~\ref{alg:error-correction}.} Denote by $\hat{\delta}_{ij}$ and $\hat{e}_{ij}$ the estimates for $\delta_{ij}$ and $e_{ij}$. A general equation system assuming that all synchronization sessions are faulty is
\begin{equation}
\left\{
\begin{array}{ll}
\hat{\delta}_{j0} + \hat{e}_{j0} = \widetilde{\delta}_{j0}, & \forall n_j \leftrightarrow n_0 \in E, j \neq 0; \\
\hat{\delta}_{i0} - \hat{\delta}_{j0} + \hat{e}_{ij} = \widetilde{\delta}_{ij}, & \forall n_i \leftrightarrow n_j \in E, i \neq 0, j \neq 0.
\end{array}
\right.
\label{eq:error-correction}
\end{equation}
The variables to be solved are the unknowns $\{ \hat{\delta}_{j0}| \forall j \in [1, N-1] \}$ and {\blue $\{\hat{e}_{ij} | \forall n_i \leftrightarrow n_j \in E \}$}, where $\hat{\delta}_{j0}$ is the estimated clock offset between $n_j$ and the reference node $n_0$; $\hat{e}_{ij}$ is the estimated p2p clock synchronization fault between $n_i$ and $n_j$ {\blue if they have performed the p2p synchronization session in the current synchronization round. Note that regardless of $E$, we aim at solving the clock offset between every node and the reference node. If there are too few edges in $E$, Eq.~(\ref{eq:error-correction}) may have infinite number of solutions. Our resilience definition in \sect\ref{subsec:problem} accounts for this situation.}

{\blue Let $k$ denote the assumed number of faults, which can be different from the actual number of faults. {\yellow The scattering of the $k$ assumed faults on the $|E|$ p2p synchronization sessions is called {\em distribution} of the assumed faults.} As shown in Algorithm~\ref{alg:error-correction}, the NCS algorithm starts by assuming there are no faults (i.e., $k \leftarrow 0$). In each iteration of the algorithm that increases $k$ by one, the algorithm solves the variants of Eq.~(\ref{eq:error-correction}) that capture all ${|E| \choose k}$ possible distributions of the $k$ assumed faulty p2p synchronization sessions among all the $|E|$ p2p synchronization sessions. Specifically, a variant of Eq.~(\ref{eq:error-correction}) is generated by keeping $k$ estimated p2p synchronization faults (i.e., $\hat{e}_{j0}$ or $\hat{e}_{ij}$) in Eq.~(\ref{eq:error-correction}) and removing other estimated p2p synchronization faults.
Once a solution is found, Algorithm~\ref{alg:error-correction} returns the estimates $\{ \hat{\delta}_{j0}| \forall j \in [1, N-1] \}$ and $\{\hat{e}_{ij} | \forall n_i \leftrightarrow n_j \in E \}$. If $\hat{\delta}_{j0} = \delta_{j0}$, $\forall j \in [1, N-1]$, we say Algorithm~\ref{alg:error-correction} can correct the faults.}

Algorithm~\ref{alg:error-correction} requires neither the {\blue actual} number nor the {\blue actual} distribution of the p2p synchronization faults. Whether it can correct the faults and how many faults it can {\blue correct} will be the focus of this paper. Algorithm~\ref{alg:error-correction} is a centralized algorithm executed on the central node. The time complexity of the $k$th step of Algorithm~\ref{alg:error-correction} is $O\left( {|E| \choose k} \right)$. Thus, the time complexity upper bound of Algorithm~\ref{alg:error-correction} is $O \left(\sum_{k=0}^{|E|}{|E| \choose k}\right)$ $= O\left(2^{|E|}\right)$. In \sect\ref{subsec:comparision}, we further show that the time complexity lower bound of Algorithm~\ref{alg:error-correction} is $\Omega\left(\sqrt{2}^N\right)$ for complete NCS graphs. Therefore, Algorithm~\ref{alg:error-correction} has an exponential time complexity. In \sect\ref{subsec:comparision}, based on a graph-theoretic analysis, we will develop a fast NCS algorithm with a time complexity of $O(N^3)$ that provides the same fault correction capability as Algorithm~\ref{alg:error-correction}. With the fast algorithm, the centralized NCS is scalable to network size. The resilience of the centralized NCS provides important baseline understanding on the resilience of synchronizing a network of nodes.

\subsection{Problem Statement}
\label{subsec:problem}

Let $\mathbb{Z}_{\ge 0}$ denote the set of non-negative integers.

\begin{definition}[$K$-resilience]
  Let $K \in \mathbb{Z}_{\ge 0}$ denote the number of faulty p2p synchronization sessions among a total of $|E|$ sessions in an NCS graph $G = (V, E)$. The network with the NCS graph $G$ is $K$-resilient if Algorithm~\ref{alg:error-correction} can correct \textit{any} $K$ p2p synchronization faults.\qed
  \label{def:resilience}
\end{definition}

From Algorithm~\ref{alg:error-correction}, we define the $K$-resilience condition that can be used to check whether a network {\blue with $G$} is $K$-resilient.

\begin{definition}[$K$-resilience condition]
A network with $G$ is $K$-resilient if the following conditions are satisfied:
\begin{enumerate}
\item $\forall k \in [0,K)$, the variant of Eq.~(\ref{eq:error-correction}) corresponding to \textit{any} distribution of the $K$ actual p2p synchronization faults and \textit{any} distribution of the $k$ estimated p2p synchronization faults has no solutions;
\item When $k=K$, for \textit{any} distribution of the $K$ actual p2p synchronization faults and \textit{any} distribution of the $k$ estimated p2p synchronization faults,
  \begin{enumerate}
  \item if the distribution of the $k$ estimated p2p synchronization faults is identical to the distribution of the actual faults, Eq.~(\ref{eq:error-correction}) has a unique solution;
  \item otherwise, Eq.~(\ref{eq:error-correction}) has no solutions. \qed
  \end{enumerate}
\end{enumerate}
\label{def:resilience-condition}
\end{definition}

Note that {\blue under} the condition (2)-(a) of Definition~\ref{def:resilience-condition}, {\blue if Eq.~(\ref{eq:error-correction}) has a unique solution, the solution must give the correct estimates of the clock offsets and the p2p synchronization faults, since these correct estimates form a valid solution.} Note that if the $K$-resilience condition in Definition~\ref{def:resilience-condition} is satisfied, Algorithm~\ref{alg:error-correction} must be able to correct any $K$ faults. However, when the condition (2)-(b) of Definition~\ref{def:resilience-condition} is not satisfied, Algorithm~\ref{alg:error-correction} can still correct $K$ faults with a specific distribution of the faults. This occurs when the first attempted distribution of the $K$ estimated faults happens to be identical to the actual distribution of the $K$ faults. However, in this case, the network is not $K$-resilient, because Definition~\ref{def:resilience} requires that Algorithm~\ref{alg:error-correction} can correct {\em any} $K$ faults to claim $K$-resilience. Thus, Definition~\ref{def:resilience-condition} gives a sufficient condition for Algorithm~\ref{alg:error-correction} to correct any $K$ faults; it is a sufficient and necessary condition for $K$-resilience.

{\blue Let $\mathbb{G}$ denote the infinite set of all NCS graphs. We define the following resilience bounds:}

\begin{definition}[Lower bound of maximum resilience]
  A function $f_l(G): \mathbb{G} \mapsto \mathbb{Z}_{\ge 0}$ is a lower bound of maximum resilience for a network with an NCS graph $G$ if the network is $K$-resilient for $K \le f_l(G)$.\qed
\end{definition}

\begin{definition}[Upper bound of maximum resilience]
  A function $f_u(G): \mathbb{G} \mapsto \mathbb{Z}_{\ge 0}$ is an upper bound of maximum resilience for a network with an NCS graph $G$ if the network is not $K$-resilient for $K > f_u(G)$.\qed
\end{definition}

\begin{definition}[Tight bound of maximum resilience]
  A function $f_t(G): \mathbb{G} \mapsto \mathbb{Z}_{\ge 0}$ is a tight bound of maximum resilience for a network with an NCS graph $G$ if the network is $K$-resilient for $K \le f_t(G)$ and not $K$-resilient for $K > f_t(G)$.\qed
  \label{def:tight}
\end{definition}

{\blue This paper aims at investigating the above resilience bounds under various NCS graph (e.g., complete or not) and the dual problem of what NCS graph condition can ensure a certain resilience bound.}

{\blue
\subsection{Relations with Real Clock Synchronization Systems}
\label{subsec:background}


The system model described in \sect\ref{subsec:system-model} includes several abstractions to clearly formulate the $K$-resilience concept in \sect\ref{subsec:problem} and allow us to focus on the essence of the problem. In this section, we discuss the potential deviations of the real systems from these abstractions and the impact of such deviations on our analysis in the reminder of this paper.

\subsubsection{Definition of fault}
\label{subsubsec:fault-definition-discussion}
In this paper, we focus on the faults that are caused by erroneous clock offset estimates. We do not consider other faults such as missing clock offset estimates. In \sect\ref{subsec:system-model}, any deviation of the measured clock offset from its true value is regarded as a fault. Under this rigorous definition of fault, we can describe the NCS algorithm and define the $K$-resilience without any vagueness. In real systems, a p2p clock synchronization session may have some clock offset estimation error due to inevitable random noises. The system designer often has good knowledge of these random noises (e.g., their sources and probabilistic distributions) and designs the clock synchronization mechanism to limit the resulted clock offset estimation errors to acceptable ranges. In practice, the clock offset estimation errors that are caused by unforeseen situations (e.g., hardware malfunction and packet delay attack \cite{rabadi2017taming,rfc7384,mizrahi2012game,ullmann2009delay}) and beyond the acceptable ranges can be regarded as faults. Following this principle, in this section, we discuss how to extend our formulation to address 1) a class of {\em sensing-based} clock synchronization systems and 2) other systems under more general settings.

In the sensing-based clock synchronization systems \cite{viswanathan2016exploiting,rabadi2017taming,yan2017,li2017natural}, the clock offset estimation errors follow a discrete pattern. Specifically, the error is given by $e_{ij} = \epsilon_{ij} + m_{ij} \cdot T$, where $T$ is the period of the external signal being sensed, $m_{ij}$ is an integer, and $\epsilon_{ij}$ is a random noise with magnitude much smaller than $T$. For instance, in the study \cite{viswanathan2016exploiting} that exploits powerline voltages to synchronize nodes in a city, $T$ is 20 milliseconds in a $50\,\text{Hz}$ power grid and the absolute value of $\epsilon_{ij}$ is about 0.1 milliseconds (i.e., 0.5\% of $T$). The discrete pattern is caused by abnormal noises of the used external signals and some integer nature of the clock synchronization algorithms to leverage on the periodicity of the external signals. For these systems, the p2p synchronization sessions with $m_{ij} \neq 0$ can be regarded as faulty sessions. Due to the random noises $\epsilon_{ij}$, Eq.~(\ref{eq:error-correction}) generally has no exact solutions even when there are no faulty p2p synchronization sessions (i.e., $K=0$) and the considered $k=0$. Instead, a candidate solution to Eq.~(\ref{eq:error-correction}) can be obtained by minimizing the following overall residual:
\begin{equation}
  \sum_{\substack{\forall n_j \leftrightarrow n_0 \in E \\ \forall j \neq 0}} \!\! \left| \hat{\delta}_{j0} + \hat{e}_{j0} - \widetilde{\delta}_{j0} \right|^2 + \!\! \sum_{\substack{\forall n_i \leftrightarrow n_j \in E \\ \forall i \neq 0, \forall j \neq 0}} \left| \hat{\delta}_{i0} - \hat{\delta}_{j0} + \hat{e}_{ij} - \widetilde{\delta}_{ij} \right|^2.
\label{eq:1}
\end{equation}
The candidate solution can be substituted into each equation in Eq.~(\ref{eq:error-correction}) to check if the absolute residual exceeds some threshold set according to the distribution of the random noise $\epsilon_{ij}$. For example, we can set one millisecond for the system in \cite{viswanathan2016exploiting}. If every absolute residual does not exceed the threshold, we view the candidate solution as a valid solution to Eq.~(\ref{eq:error-correction}) in Line~\ref{line:solve} of Algorithm~\ref{alg:error-correction}. {\gray The integer programming formulation in Eq.~(2) exploits the discrete pattern of the synchronization faults. It reduces the impact of random synchronization errors on the accuracy of determining whether Eq.~(1) has a solution.}

For systems with a more general error pattern of $e_{ij} = \epsilon_{ij} + x_{ij} \cdot F_{ij}$ where $x_{ij}$ is 0 or 1 and $F_{ij}$ is an arbitrary number beyond the range of $\epsilon_{ij}$, the synchronization sessions with $x_{ij} = 1$ can be regarded as faulty sessions. The residual minimization and candidate solution checking approaches described above can be applied as well. {\yellow In \sect\ref{subsec:simulations}, we will present a set of simulation results that consider the general error pattern.}

Despite the above variations to address acceptable clock offset estimation errors, our abstracted formulation in \sect\ref{subsec:system-model} and \sect\ref{subsec:problem} capture the essence of the problem. The analysis based on the formulation will provide insightful understanding regarding the fault resilience of the NCS mechanism in Algorithm~\ref{alg:error-correction}.

\subsubsection{Time-invariant clock offset}
In \sect\ref{subsec:system-model}, we assume that the clock offset $\delta_{ij}$ is time-invariant. In practice, the clock offset $\delta_{ij}$ can be time-varying because the clocks of $n_i$ and $n_j$ may advance at different speeds.
However, the change of $\delta_{ij}$ during a p2p synchronization session is often negligible compared with the clock offset estimation errors of successful p2p synchronization sessions. In most clock synchronization systems, a p2p synchronization session takes a short time (e.g., tens of milliseconds in NTP, PTP, and sensing-based clock synchronization such as \cite{rabadi2017taming}). Typical crystal oscillators found in microcontrollers and personal computers have drift rates of 30 to 50 parts-per-million (ppm) \cite{hao2011}. Thus, the change of the clock offset during a synchronization session of 100 milliseconds is at most 5 microseconds only, whereas the clock offset estimation errors of successful synchronization sessions are at sub-millisecond \cite{viswanathan2016exploiting,rabadi2017taming} or milliseconds levels \cite{yan2017}. Thus, the small variation of the clock offset during a synchronization session can be viewed as a nearly negligible part of the clock offset estimation error, where the latter is further much smaller than the synchronization faults as discussed in \sect\ref{subsubsec:fault-definition-discussion}. Therefore, we can safely ignore the variation of clock offset in studying the resilience of NCS against synchronization faults.}

\section{\blue $K$-Resilience Analysis for Example Networks}
\label{sec:analysis}

{\blue In this section, we present the vectorization of Eq.~(\ref{eq:error-correction}) to facilitate our analysis (\sect\ref{subsec:vectorization}) and then analyze the $K$-resilience for several small-scale networks with complete NCS graphs (\sect\ref{subsec:manual-check}). The analysis illustrates the challenges in the general analysis of the $K$-resilience for any network, but also provides guiding insights. Lastly, we provide a set of simulation results to show the impact of non-faulty synchronization errors on the NCS algorithm (\sect\ref{subsec:simulations}).}

\subsection{Vectorization}
\label{subsec:vectorization}

We vectorize the representation of Eq.~(\ref{eq:error-correction}) that is solved by Line~\ref{line:solve} of Algorithm~\ref{alg:error-correction}. Define $\hat{\vecg{\delta}} \in \mathbb{R}^{N-1}$ composed of all clock offset estimates, i.e., $\hat{\vecg{\delta}} = \left(\hat{\delta}_{10}, \hat{\delta}_{20}, \ldots, \hat{\delta}_{(N-1)0}\right)^\intercal$. Define $\hat{\vec{e}} \in \mathbb{R}^k$ composed of the $k$ p2p synchronization fault estimates. Eq.~(\ref{eq:error-correction}) can be rewritten as $
  \left( \mat{A}_1 \mat{A}_2 \right)
  \left(
  \begin{array}{c}
    \hat{\vecg{\delta}} \\
    \hat{\vec{e}}
  \end{array}
\right) = \vec{b}
$,
where {\blue $\mat{A}_1 \in \mathbb{R}^{|E| \times (N-1)}$ and $\mat{A}_2 \in \mathbb{R}^{|E| \times k}$} are two matrices composed of $-1$, $0$, and $1$ containing coefficients corresponding to $\hat{\delta}_{\cdot 0}$ and $\hat{e}_{\cdot\cdot}$, respectively;
the vector {\blue $\vec{b} \in \mathbb{R}^{|E|}$} consists of all the measured clock offsets. To simplify notation, we define $\mat{A} = \left( \mat{A}_1 \mat{A}_2 \right)$ and $\vec{x} = \left( \begin{array}{c}
    \hat{\vecg{\delta}} \\
    \hat{\vec{e}}
  \end{array} \right)$. {\blue The $\mat{A}\vec{x}=\vec{b}$ is called {\em NCS equation system}.} From the Rouch\'{e}-Capelli theorem \cite{Shafarevich2012}, the necessary and sufficient condition that $\mat{A}\vec{x} = \vec{b}$ has no solutions is $\mathrm{rank}(\mat{A} | \vec{b}) \neq \mathrm{rank}(\mat{A})$, where $\mat{A} | \vec{b}$ is the augmented matrix.

\subsection{\blue $K$-Resilience of Small-Scale Networks}
\label{subsec:manual-check}

This section presents the analysis on the $K$-resilience of {\blue several small-scale networks} with complete NCS graphs.

\begin{proposition}
  {\blue A 3-node network with a complete NCS graph is not 1-resilient.}
  \label{exmp:1}
\end{proposition}
\begin{proof}
  Consider a case where the p2p synchronization session $n_1 \leftrightarrow n_2$ is faulty. When $k=0$ in Algorithm~\ref{alg:error-correction}, the vectorized equation system in Eq.~(\ref{eq:error-correction}) is
  \begin{equation*}
    \begin{array}{cccc}
      \left(
        \begin{array}{cc}
          1 & 0 \\
          0 & 1 \\
          -1 & 1
        \end{array}
      \right)
      &
      \left(
        \begin{array}{c}
          \hat{\delta}_{10} \\
          \hat{\delta}_{20}
        \end{array}
      \right)
      &
      =
      &
      \left(
        \begin{array}{c}
          \delta_{10} \\
          \delta_{20} \\
          \delta_{20} - \delta_{10} + e_{21}
        \end{array}
      \right).
      \\
      \Uparrow & \Uparrow & & \Uparrow \\
      \mat{A} & \vec{x} &  & \vec{b}
    \end{array}
  \end{equation*}
  Note that $\mat{A}_2$ and $\hat{\vec{e}}$ are empty. With $e_{21} \neq 0$, Gaussian elimination shows that $\mathrm{rank}(\mat{A} | \vec{b}) \neq \mathrm{rank}(\mat{A})$. Thus, the equation system has no solutions and Algorithm~\ref{alg:error-correction} will move on to the case of $k=1$. The algorithm will attempt to test all the {\blue ${|E| \choose k}={3 \choose 1}=3$} possible cases of a single faulty p2p synchronization session. For instance, when the algorithm assumes that $n_0 \leftrightarrow n_1$ is faulty, the NCS equation system is
    \begin{equation*}
    \left(
      \begin{array}{ccc}
        1 & 0 & 1\\
        0 & 1 & 0 \\
        -1 & 1 & 0
      \end{array}
    \right)
    \left(
      \begin{array}{c}
        \hat{\delta}_{10} \\
        \hat{\delta}_{20} \\
        \hat{e}_{10}
      \end{array}
    \right)
    =
    \left(
      \begin{array}{c}
        \delta_{10} \\
        \delta_{20} \\
        \delta_{20} - \delta_{10} + e_{21}
      \end{array}
    \right).
  \end{equation*}
  With $e_{21} \neq 0$, we have $\mathrm{rank}(\mat{A}|\vec{b}) = \mathrm{rank}(\mat{A})$ and $\mat{A}$ has full column rank. Thus, the NCS equation system has a unique solution. Therefore, the condition (2)-(b) of Definition~\ref{def:resilience-condition} is not satisfied and the network is not 1-resilient. The unique solution is $\{ \hat{\delta}_{10} = \delta_{10} - e_{21}, \hat{\delta}_{20} = \delta_{20}, \hat{e}_{10} = e_{21} \}$, which gives wrong clock offset estimates.
\end{proof}

\begin{proposition}
  {\blue A 4-node network with a complete NCS graph is 1-resilient.}
  \label{exmp:2}
\end{proposition}

We provide a sketch of the proof as follows instead of a complete proof for presentation conciseness. In fact, this proposition is a corollary of Theorem~\ref{thm:1} with a complete proof in \sect\ref{subsec:tight-bound}. Thus, the omission of the complete proof here does not cause loss of rigor. Consider a case where the p2p synchronization session $n_0 \leftrightarrow n_2$ is faulty. When $k=0$ in Algorithm~\ref{alg:error-correction},
similar to Proposition~\ref{exmp:1}, the NCS equation system has no solutions and Algorithm~\ref{alg:error-correction} will move on to the case of $k=1$. The algorithm will test all the {\blue ${|E| \choose k} = {6 \choose 1}=6$} possible cases of a single faulty p2p synchronization session. For instance, when the algorithm assumes $n_0 \leftrightarrow n_1$ is faulty, the NCS equation system is
\begin{equation}
\left(
\begin{array}{cccc}
1 & 0 & 0 & 1 \\
0 & 1 & 0 & 0 \\
0 & 0 & 1 & 0 \\
-1 & 1 & 0 & 0 \\
0 & -1 & 1 & 0 \\
-1 & 0 & 1 & 0
\end{array}
\right)
\left(
\begin{array}{c}
\hat{\delta}_{10}\\
\hat{\delta}_{20}\\
\hat{\delta}_{30}\\
\hat{e}_{10}
\end{array}
\right)
=
\left(
\begin{array}{c}
\delta_{10} \\
\delta_{20} + e_{20} \\
\delta_{30} \\
\delta_{20} - \delta_{10} \\
\delta_{30} - \delta_{20} \\
\delta_{30} - \delta_{10}
\end{array}
\right).
\label{eq:4node-system0}
\end{equation}
As $\mathrm{rank}(\mat{A} | \vec{b}) \neq \mathrm{rank}(\mat{A})$, the NCS equation system has no solutions. An exhaustive check shows that, only when the algorithm assumes the synchronization session between $n_0$ and $n_2$ is faulty, the NCS equation system has a unique solution (i.e., $\mathrm{rank}(\mat{A} | \vec{b}) = \mathrm{rank}(\mat{A})$ and $\mat{A}$ has full column rank). Thus, the algorithm can correct the fault. In fact, it can be verified that, for the {\blue complete 4-node NCS graph}, no matter which p2p synchronization session is faulty, the algorithm can correct the fault. Therefore, the 4-node system is 1-resilient.

\begin{proposition}
  {\blue A 4-node network with a complete NCS graph is not 2-resilient.}
  \label{exmp:3}
\end{proposition}
\begin{proof}
  Consider the 4-node network with two faulty p2p synchronization sessions: $n_0 \leftrightarrow n_1$ and $n_0 \leftrightarrow n_2$. When $k=0$, the equation system has no solutions. When $k=1$, consider a case where $n_0 \leftrightarrow n_3$ is assumed to be faulty by the algorithm. The NCS equation system is
  \begin{equation}
    \left(
      \begin{array}{cccc}
        1 & 0 & 0 & 0 \\
        0 & 1 & 0 & 0 \\
        0 & 0 & 1 & 1 \\
        -1 & 1 & 0 & 0 \\
        -1 & 0 & 1 & 0 \\
        0 & -1 & 1 & 0
      \end{array}
    \right)
    \left(
      \begin{array}{c}
        \hat{\delta}_{10} \\
        \hat{\delta}_{20} \\
        \hat{\delta}_{30} \\
        \hat{e}_{30}
      \end{array}
    \right)
    =
    \left(
      \begin{array}{c}
        \delta_{10} + e_{10} \\
        \delta_{20} + e_{20} \\
        \delta_{30} \\
        \delta_{20} - \delta_{10} \\
        \delta_{30} - \delta_{10} \\
        \delta_{30} - \delta_{20}
      \end{array}
    \right).
    \label{eq:4node-system1}
  \end{equation}
  If $e_{10} \neq e_{20}$, $\mathrm{rank}(\mat{A}|\vec{b}) \neq \mathrm{rank}(\mat{A})$ and the equation system has no solutions. However, if $e_{10} = e_{20}$, $\mathrm{rank}(\mat{A} | \vec{b}) = \mathrm{rank}(\mat{A})$ and $\mat{A}$ has full column rank; the equation system has a unique solution of $\{\hat{\delta}_{10} = \delta_{10} + e_{10}, \hat{\delta}_{20} = \delta_{20} + e_{10},  \hat{\delta}_{30} = \delta_{30} + e_{10}, \hat{e}_{30} = - e_{10}\}$, which gives wrong clock offset estimates. Although this counterexample against the 4-node network's 2-resilience is obtained under a certain condition of $e_{10} = e_{20}$, we can conclude that the 4-node network is not 2-resilient.
  \end{proof}

To gain more insights, we also analyze a case of $k=2$ with $n_0 \leftrightarrow n_1$ and $n_0 \leftrightarrow n_3$ assumed to be faulty by the algorithm. The NCS equation system is
  \begin{equation}
      \left(
      \begin{array}{ccccc}
        1 & 0 & 0 & 1 & 0 \\
        0 & 1 & 0 & 0 & 0 \\
        0 & 0 & 1 & 0 & 1 \\
        -1 & 1 & 0 & 0 & 0 \\
        -1 & 0 & 1 & 0 & 0 \\
        0 & -1 & 1 & 0 & 0
      \end{array}
    \right)
    \left(
      \begin{array}{c}
        \hat{\delta}_{10} \\
        \hat{\delta}_{20} \\
        \hat{\delta}_{30} \\
        \hat{e}_{10} \\
        \hat{e}_{30}
      \end{array}
    \right)
    =
    \left(
      \begin{array}{c}
        \delta_{10} + e_{10} \\
        \delta_{20} + e_{20} \\
        \delta_{30} \\
        \delta_{20} - \delta_{10} \\
        \delta_{30} - \delta_{10} \\
        \delta_{30} - \delta_{20}
      \end{array}
    \right).
    \label{eq:4node-system2}
\end{equation}
As $\mathrm{rank}(\mat{A} | \vec{b}) = \mathrm{rank}(\mat{A})$ and $\mat{A}$ has full column rank, the equation system has a unique solution, which violates the 2-resilience condition. In fact, the equation system has a unique solution that  gives wrong clock offset estimates and does not require any relationship between $e_{10}$ and $e_{20}$. This solution is $\{\hat{\delta}_{10} = \delta_{10} + e_{20}, \hat{\delta}_{20} = \delta_{20} + e_{20}, \hat{\delta}_{30} = \delta_{30} + e_{20}, \hat{e}_{10} = e_{10} - e_{20}, \hat{e}_{30} = -e_{20}\}$.

\begin{proposition}
  {\blue A 5-node network with a complete NCS graph is 1-resilient.}
  \label{exmp:4}
\end{proposition}

We provide a sketch of the proof as follows instead of a complete proof due to space limit. This proposition is in fact a corollary of Theorem~\ref{thm:1} with a complete proof. Thus, the omission here does not cause loss of rigor. Consider a 5-node network with one p2p synchronization fault. The resilience is independent from how we name the nodes. We name the two nodes involved in the faulty synchronization session as $n_0$ and $n_1$. An exhaustive check over all the {\blue ${|E| \choose k} = {10 \choose 1}$ = 10} possible cases for a single assumed faulty synchronization session shows that the 1-resilience condition is satisfied. Thus, the 5-node network is 1-resilient.

\begin{proposition}
  {\blue A 5-node network with a complete NCS graph is not 2-resilient.}
  \label{exmp:5}
\end{proposition}
\begin{proof}
  We consider a 5-node network, in which (i) the p2p synchronization sessions $n_0 \leftrightarrow n_1$ and $n_1 \leftrightarrow n_4$ are faulty and (ii) the p2p synchronization sessions $n_1 \leftrightarrow n_2$ and $n_1 \leftrightarrow n_3$ are assumed by the algorithm to be faulty. The NCS equation system is
  \begin{equation}
    \left(
      \begin{array}{cccccc}
        1 & 0 & 0 & 0 & 0 & 0 \\
        0 & 1 & 0 & 0 & 0 & 0 \\
        0 & 0 & 1 & 0 & 0 & 0 \\
        0 & 0 & 0 & 1 & 0 & 0 \\
        -1 & 1 & 0 & 0 & 1 & 0 \\
        -1 & 0 & 1 & 0 & 0 & 1 \\
        -1 & 0 & 0 & 1 & 0 & 0 \\
        0 & -1 & 1 & 0 & 0 & 0 \\
        0 & -1 & 0 & 1 & 0 & 0 \\
        0 & 0 & -1 & 1 & 0 & 0
      \end{array}
    \right)
    \left(
      \begin{array}{c}
        \hat{\delta}_{10} \\
        \hat{\delta}_{20} \\
        \hat{\delta}_{30} \\
        \hat{\delta}_{40} \\
        \hat{e}_{21} \\
        \hat{e}_{31}
      \end{array}
    \right)
    =
    \left(
      \begin{array}{c}
        \delta_{10} + e_{10} \\
        \delta_{20} \\
        \delta_{30} \\
        \delta_{40} \\
        \delta_{20} - \delta_{10} \\
        \delta_{30} - \delta_{10} \\
        \delta_{40} - \delta_{10} + e_{41} \\
        \delta_{30} - \delta_{20} \\
        \delta_{40} - \delta_{20} \\
        \delta_{40} - \delta_{30}
      \end{array}
    \right).
    \label{eq:5node-system}
  \end{equation}
If $e_{10} = -e_{41}$, the equation system has a unique solution of $\{\hat{\delta}_{10} = \delta_{10} + e_{10}, \hat{\delta}_{20} = \delta_{20}, \hat{\delta}_{30} = \delta_{30}, \hat{\delta}_{40} = \delta_{40}, \hat{e}_{21} = e_{10}, \hat{e}_{31} = e_{10} \}$, which violates the resilience condition. Thus, a 5-node network is not 2-resilient.
\end{proof}

{\blue In the proofs of Propositions~\ref{exmp:1}, \ref{exmp:3}, and \ref{exmp:5}, we adopt an approach of enumerating counterexamples to prove that a network is not $K$-resilient. In the proofs of Propositions~\ref{exmp:3} and \ref{exmp:5}, if the actual faults satisfy certain conditions, the rank of $\mat{A}|\vec{b}$ may change, presenting a pitfall to the approach of enumerating counterexamples. This is a challenge in pursuing the general analysis for $K$-resilience. {\yellow To address this challenge, in} \sect\ref{sec:fully}, we will introduce a fault-free NCS subgraph method to analyze the tight bound of complete NCS graphs.}

\subsection{Simulations with Non-Faculty Synchronization Errors}
\label{subsec:simulations}

We conduct simulations to evaluate the impact of the non-faculty synchronization errors discussed in \sect\ref{subsubsec:fault-definition-discussion} on the performance of Algorithm~\ref{alg:error-correction}. The simulations are for the small-scale example network topologies analyzed in \sect\ref{subsec:manual-check}. The p2p synchronization session follows a general error pattern of $e = \epsilon + x \cdot F$, where $\epsilon$ is a Gaussian noise following the standard normal distribution, the absolute value of the fault $F$ is uniformly distributed within the range of [2, 8], the binary cofficient $x$ is randomly sampled from $\{0, 1\}$. The synchronization error with $x=1$ is viewed as a fault. The clock offset of each node with respect to the reference node is randomly and uniformly sampled from $[-10, 10]$. We use the least squares approach in Eq.~(\ref{eq:1}) to generate the candidate solution of the NCS equation system. As discussed in \sect\ref{subsubsec:fault-definition-discussion}, we apply a threshold of $\eta=2$ to check the residual of each equation of the NCS equation system to decide whether a candidate solution is a valid solution. With the setting of $\eta=2$, the probability of misclassifying a non-faulty synchronization error as a fault is $\Pr(|\epsilon| \ge 2) = 0.046$. For each network, we simulate a number of cases with different numbers of faults. For each case, we report  the mean-square error (MSE) of the estimated clock offsets of all the nodes and whether the distribution of the estimated p2p synchronization faults is identical to the distribution of the actual faults.

Tables~\ref{table:3node},~\ref{table:4node}, and~\ref{table:5node} show the simulation results for the 3-node, 4-node, and 5-node networks, respectively. From Table~1, we can see that Algorithm~1 can opportunistically correct one fault, in which the distributions of the estimated faults and actual faults are identical and the MSE is small. When the distributions of the estimated p2p synchronization faults and the actual faults are not identical, the MSE of the estimated clock offsets is large. This means that Algorithm~\ref{alg:error-correction} cannot correct the faults. From Table~1, Algorithm~\ref{alg:error-correction} cannot correct more than one fault. This result is consistent with Proposition~1. Note that we have explained in \sect\ref{subsec:problem} that, for a network that is not $K$-resilient, Algorithm~\ref{alg:error-correction} may be able to correct $K$ faults with a specific distribution that happens to be identical to the first attempted fault distribution in Algorithm~\ref{alg:error-correction}. From Tables~2 and 3, we can see that Algorithm~\ref{alg:error-correction} can always correct one fault and opportunistically correct two faults. This result is consistent with Propositions~2,~3,~4, and~5. From this set of simulation results, we can see that our analytic results provide good understanding for the scenarios with non-faulty synchronization errors.

\begin{table}[ht]
	\caption{NCS results of a $3$-node network with non-faulty synchronization errors.} 
	\centering 
		\begin{minipage}{\columnwidth}
		\centering
	\begin{tabular}{c c c c} 
		\hline\hline 
		Case & Number of faults & Identical?* & MSE of estimated clock offsets \\ [1ex] 
		\hline 
		1 & 1 & yes & 2.8721 \\ 
		2 & 1 & no & 14.5525  \\
		3 & 1 & no & 23.4256 \\
		4 & 1 & no & 37.2651  \\
		5 & 2 & no  & 34.5164  \\
		6 & 2 & no  & 53.2659   \\
		7 & 2 & no  &  67.5983  \\
		8 & 3 & no  &  79.3514  \\ [1ex] 
		\hline 
	\end{tabular}

        {\footnotesize *This column indicates whether the distributions of estimated faults and actual faults are identical.}
\end{minipage}
    \centering
    
	\label{table:3node} 
\end{table}

\begin{table}[ht]
	\caption{NCS results of a $4$-node network with non-faulty synchronization errors.} 
	\centering 
	\begin{minipage}{\columnwidth}
		\centering
	\begin{tabular}{c c c c} 
		\hline\hline 
		Case & Number of faults & Identical?* & MSE of estimated clock offsets \\ [0.5ex] 
		\hline 
		1 & 1 & yes & 1.1329  \\ 
		2 & 1 & yes & 1.9569  \\
		3 & 1 & yes & 0.5390  \\
		4 & 1 & yes & 2.1312  \\
		5 & 2 & yes & 2.0863   \\
		6 & 2 & no  & 13.3885   \\
		7 & 2 & no  &  40.8532   \\
		8 & 3 & no  &  58.8646  \\ [1ex] 
		\hline 
	\end{tabular}

        {\footnotesize *This column indicates whether the distributions of estimated faults and actual faults are identical.}
\end{minipage}
	\label{table:4node} 
\end{table}

\begin{table}[ht]
	\caption{NCS results of a $4$-node network with non-faulty synchronization errors.} 
	\centering 
		\begin{minipage}{\columnwidth}
		\centering
	\begin{tabular}{c c c c} 
		\hline\hline 
		Case& Number of faults & Identical?*  & MSE of estimated clock offsets\\ [0.5ex] 
		\hline 
		1 & 1 & yes & 2.1296  \\ 
		2 & 1 & yes & 1.0586  \\
		3 & 1 & yes & 2.2766  \\
		4 & 2 & yes & 2.5198  \\
		5 & 2 & no  & 13.5984   \\
		6 & 2 & no  & 10.5354   \\
		7 & 3 & no  &19.8177   \\
		8 & 3 & no  &22.5147  \\ [1ex] 
		\hline 
	\end{tabular}

        {\footnotesize *This column indicates whether the distributions of estimated faults and actual faults are identical.}
\end{minipage}
	\label{table:5node} 
      \end{table}

\section{Tight Bound of Maximum Resilience of Any Network with Complete NCS Graph}
\label{sec:fully}

In this section, our analysis shows that the tight bound of maximum resilience of any $N$-node network with a complete NCS graph is $f_t(N) = \left\lfloor \frac{N}{2} \right\rfloor - 1$. Note that in this section we change the notation $f_t(G)$ defined in Definition~\ref{def:tight} to $f_t(N)$, because the complete NCS graph $G$ solely depends on $N$. In what follows, we introduce the {\em fault-free NCS subgraph} (\sect\ref{subsec:ff-graph}) and prove two lemmas (\sect\ref{subsec:tight-lemmas}). The lemmas will be used to prove the tight bound of maximum resilience (\sect\ref{subsec:tight-bound}).

\subsection{Fault-Free NCS Subgraph}
\label{subsec:ff-graph}

For a certain distribution of the estimated p2p synchronization faults among the $|E|$ sessions, we retain all the equations in Eq.~(\ref{eq:error-correction}) that contain neither estimated fault $\hat{e}_{ij}$ nor actual fault $e_{ij}$ to generate an equation subsystem $\mat{A}_s \hat{\vecg{\delta}} = \vec{b}_s$.
This equation subsystem corresponds to a fault-free NCS subgraph $G_s = (V, E_s)$, where each edge in $E_s$ represents a p2p synchronization session associated with neither estimated nor actual synchronization fault. The $G_s$ is a subgraph of the original complete NCS graph $G$.

For instance, to generate the $\mat{A}_s \hat{\vecg{\delta}} = \vec{b}_s$ of Eq.~(\ref{eq:5node-system}), we can remove the rows and columns of Eq.~(\ref{eq:5node-system}) as follows:
\begin{equation*}
      \left(
      \begin{array}{cccccc}
        \hcancel{1} & \hcancel{0} & \hcancel{0} & \hcancel{1} & \hcancel{0} \\
        \hcancel{0} & \hcancel{1} & \hcancel{0} & \hcancel{0} & \hcancel{0} \\
        \hcancel{0} & \hcancel{0} & \hcancel{1} & \hcancel{0} & \hcancel{1} \\
        -1 & 1 & 0 & \hcancel{0} & \hcancel{0} \\
        -1 & 0 & 1 & \hcancel{0} & \hcancel{0} \\
        0 & -1 & 1 & \hcancel{0} & \hcancel{0}
      \end{array}
    \right)
    \left(
      \begin{array}{c}
        \hat{\delta}_{10} \\
        \hat{\delta}_{20} \\
        \hat{\delta}_{30} \\
        \hcancel{\hat{e}_{10}} \\
        \hcancel{\hat{e}_{30}}
      \end{array}
    \right)
    =
    \left(
      \begin{array}{c}
        \hcancel{\delta_{10} + e_{10}} \\
        \hcancel{\delta_{20} + e_{20}} \\
        \hcancel{\delta_{30}} \\
        \delta_{20} - \delta_{10} \\
        \delta_{30} - \delta_{10} \\
        \delta_{30} - \delta_{20}
      \end{array}
    \right).
\end{equation*}
The first and the third rows of $\mat{A}$ and $\vec{b}$ are removed because they involve estimated faults $\hat{e}_{10}$ and $\hat{e}_{30}$. Specifically, the fourth element of $\mat{A}$'s first row that corresponds to $\hat{e}_{10}$ is 1; the last element of $\mat{A}$'s third row that corresponds to $\hat{e}_{30}$ is 1. The second row of $\mat{A}$ is removed because it involves the actual fault $e_{20}$  from the second row of $\vec{b}$. The last two columns of $\mat{A}$ are removed because we no longer have $\hat{e}_{10}$ and $\hat{e}_{30}$. The remainders form $\mat{A}_s \hat{\vecg{\delta}} = \vec{b}_s$, i.e.,
\begin{equation*}
\left(
\begin{array}{ccc}
-1 & 1 & 0\\
-1 & 0 & 1 \\
 0 & -1 & 1
\end{array}
\right)
\left(
\begin{array}{c}
\hat{\delta}_{10} \\
\hat{\delta}_{20} \\
\hat{\delta}_{30}
\end{array}
\right)
=
\left(
\begin{array}{c}
\delta_{20} - \delta_{10} \\
\delta_{30} - \delta_{10} \\
\delta_{30} - \delta_{20} 
\end{array}
\right).
\end{equation*}

\subsection{The Lemmas}
\label{subsec:tight-lemmas}

Following the convention of graph theory, we say an undirected graph is {\em connected} when there is a path between every pair of vertices. We have the following lemmas.

\begin{lemma}
  For a complete NCS graph $G = (V, E)$ and a certain distribution of the estimated p2p synchronization faults, if the fault-free NCS subgraph $G_s$ is connected, the NCS equation system $\mat{A}\vec{x}=\vec{b}$ has at most one solution.
  \label{lemma:graph1}
\end{lemma}
\begin{proof}
\label{proof:fully1}
Since $G_s$ is connected, we can find a traversal of $G_s$ starting from $n_0$ and ending at any node $n_i$, which is represented by a list $\langle n_0, n_{w_1}, n_{w_2}, \ldots, n_{w_p}, n_i \rangle$.
Note that, in the above list, two different symbols $n_{w_x}$ and $n_{w_y}$ may refer to the same node in the network. We can formulate a system of equations along the above traversal, where each equation corresponds to an edge connecting two consecutive nodes in the traversal. The equation system consists of
$\hat{\delta}_{w_10}=\delta_{w_10}$, $\hat{\delta}_{w_10}-\hat{\delta}_{w_20}=\delta_{w_10}-\delta_{w_20}$, $\hat{\delta}_{w_20}-\hat{\delta}_{w_30}=\delta_{w_20}-\delta_{w_30}$, $\ldots$, $\hat{\delta}_{w_{p-1}0}-\hat{\delta}_{w_p0}=\delta_{w_{p-1}0}-\delta_{w_p0}$, $\hat{\delta}_{w_p0}-\hat{\delta}_{w_i0}=\delta_{w_p0}-\delta_{w_i0}$. By substituting the solution of the previous equation to the next equation in the above chain of equations, we have a unique solution of $\hat{\vecg{\delta}}=\vecg{\delta}$.

We substitute the solution $\hat{\vecg{\delta}}=\vecg{\delta}$ into the original equation system $\mat{A}\vec{x}=\vec{b}$ to solve the remaining unknown variables $\hat{\vec{e}}$. There are three cases for the equations in $\mat{A}\vec{x}=\vec{b}$ but not in $\mat{A}_s \hat{\vecg{\delta}} = \vec{b}_s$:
\begin{enumerate}
\item For an actually faulty edge $n_i \leftrightarrow n_j$ that is correctly assumed to be faulty, the equation is $\hat{\delta}_{i0}-\hat{\delta}_{j0}+\hat{e}_{ij}=\delta_{i0}-\delta_{j0}+e_{ij}$. By substituting $\hat{\delta}_{i0} = \delta_{i0}$ and $\hat{\delta}_{j0} = \delta_{j0}$ (which are from $\hat{\vecg{\delta}}=\vecg{\delta}$) into the above equation, we have $\hat{e}_{ij} = e_{ij}$.
\item For an actually non-faulty edge $n_i \leftrightarrow n_j$ that is wrongly assumed to be faulty, the equation is $\hat{\delta}_{i0}-\hat{\delta}_{j0}+\hat{e}_{ij}=\delta_{i0}-\delta_{j0}$. The solution is $\hat{e}_{ij} = 0$.
\item For an actually faulty edge $n_i \leftrightarrow n_j$ that is wrongly assumed to be non-faulty, the equation is $\hat{\delta}_{i0}-\hat{\delta}_{j0}=\delta_{i0}-\delta_{j0}+e_{ij}$. Since $\hat{\delta}_{i0} = \delta_{i0}$, $\hat{\delta}_{j0} = \delta_{j0}$, and $e_{ij} \neq 0$ (which is the given condition), the above equation does not hold.
\end{enumerate}

The $\mat{A} \vec{x} = \vec{b}$ that contains case (3) has no solution; the $\mat{A} \vec{x} = \vec{b}$ that does not contain case (3) has a unique solution that gives correct clock offset estimates. Thus, $\vec{A} \vec{x} = \vec{b}$ has at most one solution.
\end{proof}

Denote by $A \setminus B$ the relative complement of a set $B$ with respect to a set $A$, i.e., the set of elements in $A$ but not in $B$. We have the following lemma.

\begin{lemma}
  For a complete NCS graph $G = (V, E)$ and any edge subset $M \subseteq E$, a sufficient condition for {\yellow the} subgraph $G' = (V, E \setminus M)$ {\yellow to be connected} is $|M| \leq 2 \cdot \left( \left\lfloor \frac{N}{2} \right\rfloor - 1 \right)$, where $N = |V|$.
  \label{lemma:graph2}
\end{lemma}
\begin{proof}
  Let $C_1$ denote the clause of $|M| \leq 2 \cdot \left( \left\lfloor \frac{N}{2} \right\rfloor - 1 \right)$; let $C_2$ denote the clause of $G'$ is connected. From logic, we have the following equivalence: $(C_1 \Rightarrow C_2) \Leftrightarrow (\neg C_1 \Leftarrow \neg C_2)$, where $\neg$ represents negation. The clause $\neg C_1$ is $|M| > 2 \cdot \left( \left\lfloor \frac{N}{2} \right\rfloor - 1 \right)$. As $|E| = \frac{N(N-1)}{2}$, the clause $\neg C_1$ is also equivalent to $|E \setminus M| \le \frac{N(N-1)}{2} - 2 \cdot \left( \left\lfloor \frac{N}{2} \right\rfloor - 1 \right)$. The clause $\neg C_2$ is that $G'$ is disconnected. From the above reasoning, the sufficient condition to be proved is equivalent to the following: a sufficient condition for $|E \setminus M| \le \frac{N(N-1)}{2} - 2 \cdot \left( \left\lfloor \frac{N}{2} \right\rfloor - 1 \right)$ is that $G'$ is disconnected. In the following, we prove this equivalent sufficient condition.

Since $G'$ is disconnected, we assume that it has a total of $P$ partitions, where $P \ge 2$. Let $N_p \in \mathbb{Z}_{>0}$ denote the number of vertices in the $p$th partition. Thus, $\sum_{p=1}^P N_p = N$. Define $N_r = N - N_1 = \sum_{p=2}^{P} N_p$. We have
  \begin{align}
    \frac{N_r(N_r-1)}{2} &= \frac{ \left(\sum_{p=2}^{P} N_p \right) \left(\sum_{p=2}^{P} N_p - 1 \right)}{2} \nonumber \\
                         &= \sum_{p=2}^P \frac{N_p \left( N_p-1 \right)}{2} + \sum_{\forall p,q \in [2, P], p \neq q}N_pN_q \nonumber \\
                         &\geq \sum_{p=2}^P \frac{N_p(N_p-1)}{2}. \label{eq:inq-1}
  \end{align}
  As the number of edges of the $p$th partition is no greater than $\frac{N_p(N_p-1)}{2}$, we have $|E \setminus M| \le \sum_{p=1}^P \frac{N_p(N_p-1)}{2} = \frac{N_1 (N_1 - 1)}{2} + \sum_{p=2}^P \frac{N_p(N_p-1)}{2} \le \frac{N_1 (N_1 - 1)}{2} + \frac{N_r(N_r-1)}{2}$, where the last inequality follows from Eq.~(\ref{eq:inq-1}). By substituting $Nr = N-N_1$ into the above inequality, we have $|E \setminus M| \le \frac{N(N-1)}{2} + N_1 (N_1 - N)$. Note that $N_1 \in [1, N-1]$. When $N_1 = 1$ or $N_1 = N-1$, the quadratic $N_1(N_1 - N)$ achieves its maximum value of $-(N-1)$. Thus, $|E \setminus M| \le \frac{N(N-1)}{2} - (N-1)$. As $-(N-1) < -2 \cdot \left( \left\lfloor \frac{N}{2} \right\rfloor - 1 \right)$, we have $|E \setminus M| \le \frac{N(N-1)}{2} -2 \cdot \left( \left\lfloor \frac{N}{2} \right\rfloor - 1 \right)$.
\end{proof}

\subsection{Tight Bound of Maximum Resilience}
\label{subsec:tight-bound}

\begin{theorem}
  The tight bound of maximum resilience of any $N$-node network with a complete NCS graph $G=(V,E)$ is $f_t(N) = \left\lfloor \frac{N}{2} \right\rfloor - 1$.
  \label{thm:1}
\end{theorem}

\begin{proof}
  First, we prove that, if $K \le f_t(N)$, the network is $K$-resilient. Let $k$ denote the assumed number of faults in Algorithm~\ref{alg:error-correction}, where $k \le K$. For any distribution of the estimated faults, let $M$ denote the set of edges excluded from $E$ to generate the fault-free NCS subgraph $G_s$. Thus, $|M| \le k + K$. Moreover, since $k \le K \le f_t(N)$, we have $|M| \le k + K \le 2 \cdot f_t(N) = 2 \cdot \left( \left\lfloor \frac{N}{2} \right\rfloor - 1 \right)$. From Lemma~\ref{lemma:graph2}, $G_s$ is connected. From Lemma~\ref{lemma:graph1}, the NCS equation system $\mat{A}\vec{x}=\vec{b}$ has at most one solution. Now, we verify the $K$-resilience condition in Definition~\ref{def:resilience-condition} as follows:
  \begin{enumerate}
    \item When $k < K$: There must exist an actually faulty edge wrongly assumed to be non-faulty, i.e., Case (3) in the proof of Lemma~\ref{lemma:graph1}. Thus, the $\mat{A}\vec{x}=\vec{b}$ has no solution and Algorithm~\ref{alg:error-correction} will not return when $k < K$.
    \item When $k = K$: Only when the distribution of the estimated faults is correct, the $\mat{A}\vec{x}=\vec{b}$ does not encompass Case (3) in the proof of Lemma~\ref{lemma:graph1} and it must yield a solution that gives correct clock offset estimates. Otherwise, $\mat{A}\vec{x}=\vec{b}$ must encompass Case (3) in the proof of Lemma~\ref{lemma:graph1} and it has no solution.
  \end{enumerate}
Since the $K$-resilience condition holds, the network is $K$-resilient.

Second, we prove that, if $K > f_t(N)$, the network is not $K$-resilient. We prove it using an example network that is not $K$-resilient. The example network has the following two properties: (1) all actually faulty synchronization sessions and all assumed faulty synchronization sessions involve a certain node $n_i$; (2) each of the edges involving $n_i$ is either actually faulty or assumed to be faulty, or both. For this example network, $k + K \ge N - 1$, where $N-1$ is the total number of edges involving $n_i$. Moreover, as $k \le K$, we have $2K \ge k + K \ge N-1$ and $K \ge \frac{N-1}{2} > \left\lfloor \frac{N}{2} \right\rfloor - 1$. Thus, the example network satisfies $K > f_t(N)$. We now prove that this example network is not $K$-resilient. The fault-free NCS subgraph $G_s$ is disconnected and has two partitions. One of them involving all nodes except $n_i$ is a complete NCS graph without any fault. Thus, a unique partial solution that give correct clock offset estimate can be obtained for this partition. By substituting the partial solution into the original equation system $\mat{A}\vec{x} = \vec{b}$, any remaining equation that must involve $n_i$ will be in one of the following three forms: (1) $\hat{\delta}_{i0}=\delta_{i0}+e_{ik}$, (2) $\hat{\delta}_{i0}+\hat{e}_{ik}=\delta_{i0}+e_{ik}$, and (3) $\hat{\delta}_{i0}+\hat{e}_{ik}=\delta_{i0}$. If all actual faults have identical value, i.e., $e_{ik}=e$, the remaining equations will yield a solution that gives wrong clock offset estimates, in which (1) $\hat{\delta}_{i0}=\delta_{i0}+e$, (2) $\hat{e}_{ik}=0$, and (3) $\hat{e}_{ik}=-e$ that respectively correspond to the three forms. Thus, the network is not $K$-resilient.
\end{proof}

From Theorem~\ref{thm:1}, for networks with complete NCS graphs, the maximum number of correctable faults increases with $N$ in a nearly linear manner. However, the number of edges increases with $N$ quadratically. This suggests that, for networks with complete NCS graphs, the fault correction capability decreases with $N$. Thus, it is interesting to study whether we can remove edges from a complete NCS graph while maintaining $K$-resilience. To answer this question, we analyze the resilience bounds for NCS graphs that may be incomplete (\sect\ref{sec:lower-bounds}) and then analyze the minimum number of edges needed to ensure $K$-resilience (\sect\ref{sec:another}).

\section{Algorithm to Compute Tight Bound of Maximum Resilience of any Network}
\label{sec:lower-bounds}

 In this section, we study the tight bound for any NCS graphs that may be incomplete. In \sect\ref{subsec:subgraph-method}, we interpret the $K$-resilience of any NCS graph from the edge connectivity of the graph. The interpretation is mainly from the Menger's theorem \cite{bohme2001menger}. Based on the edge-connectivity interpretation, in Section~\ref{subsec:alguncomplete}, we present an algorithm to compute the tight bound of maximum resilience for any given NCS graph. Note that, different from \sect\ref{sec:fully} that gives the closed-form tight bound of maximum resilience of any complete NCS graph, the closed-form tight bound may not exist for NCS graphs that may be incomplete, because the $K$-resilience depends on the topology of the incomplete NCS graph.

\subsection{Graph-Theoretic $K$-Resilience Condition}
\label{subsec:subgraph-method}
{\yellow First, we define \textit{fault-free} path between any two nodes in a connected NCS graph. For any path between node $n_i$ and $n_j$, if every edge in the path is associated with neither estimated nor actual synchronization fault, the path is called fault-free path. Note that for a certain distribution of the estimated p2p synchronization faults among the sessions, any path is either fault-free path or non-fault-free path.}

Now, we introduce the concept of {\yellow \textit{edge-connectivity}} \cite{west1996introduction} of any graph and \textit{minimum edge cut} of any pair of nodes to extend the fault-free NCS subgraph method. In graph theory, a connected graph is \textit{$L$-edge-connected} if it remains connected {when any no greater than $L$} edges are removed from the graph. The edge-connectivity of a graph is the largest $L$ for which the graph is still $L$-edge-connected \cite{west1996introduction}. A minimum edge cut of any pair of nodes is an edge cut of the pair such that there is no other edge cut of the pair containing fewer edges. The Menger's theorem \cite{bohme2001menger} stated below will be used to analyze the resilience of any NCS graph.
\begin{theorem}
\label{thm:menger1}
  \textbf{Menger's theorem \cite{bohme2001menger}.} In a graph $G$, the size of the minimum edge cut of any pair of nodes is equal to the maximum number of disjoint paths that can be found between the  node pair. Extended to all node pairs, $G$ is $L$-edge-connected if and only if every node pair has $L$ edge-disjoint paths connecting them.
\end{theorem}
Based on Menger's theorem, we have the following theorem regarding the resilience of any NCS graph. In the proof, we use examples provided in {\em italic text} to help understanding.

\begin{theorem}
\label{thm:new1}
An NCS graph $G$ is $K$-resilient if and only if it is $(2K+1)$-edge-connected.
\end{theorem}

\begin{proof}
	Let $C_1$ denote the clause that an NCS graph $G$ is $K$-resilient; let $C_2$ denote the clause that the NCS graph is $(2K+1)$-edge-connected.
	
    \textbf{Proof of backward implication $C_2 \Rightarrow C_1$.} Assuming that the NCS graph $G$ is $(2K+1)$-edge-connected, from the Menger's theorem, there exist $2K+1$ edge-disjoint paths between any node $n_j$ and the reference node $n_0$. In the case that the network has $K$ actual faults and $k$ estimated faults where $k \leq K$, the total number of edges associated with the estimated or actual synchronization faults is less than $K+k$. Note that $K+k \leq 2K$.  From the principle of drawers, there exists at least one fault-free path among the $2K+1$ edge-disjoint paths {\yellow connecting} $n_j$ and $n_0$. We can formulate a system of equations along the above fault-free path where each equation corresponds to an edge. We denote this fault-free path as $\langle n_0, n_{w_1}, n_{w_2}, \ldots, n_{w_p}, n_j \rangle$. The equation system on the path consists of $\hat{\delta}_{w_10}=\delta_{w_10}$, $\hat{\delta}_{w_10}-\hat{\delta}_{w_20}=\delta_{w_10}-\delta_{w_20}$, $\hat{\delta}_{w_20}-\hat{\delta}_{w_30}=\delta_{w_20}-\delta_{w_30}$, $\ldots$, $\hat{\delta}_{w_{p-1}0}-\hat{\delta}_{w_p0}=\delta_{w_{p-1}0}-\delta_{w_p0}$, $\hat{\delta}_{w_p0}-\hat{\delta}_{w_j0}=\delta_{w_p0}-\delta_{w_j0}$. Similar to the proof of Lemma~\ref{lemma:graph1}, we can substitute the solution of the previous equation to the next equation in the above chain of equations {\yellow to generate the} solution of $\hat{\delta}_{j0}=\delta_{j0}$. {\yellow By repeating the above process for every non-reference node $n_j$}, we obtain a unique solution $\hat{\vecg{\delta}}=\vecg{\delta}$. We substitute the solution $\hat{\vecg{\delta}}=\vecg{\delta}$ into the original equation system $\mat{A}\vec{x}=\vec{b}$ to solve the remaining unknown variables $\hat{\vec{e}}$. As shown in the proof of {\yellow Lemma~\ref{lemma:graph1}}, when $k = K$ and the distribution of the estimated p2p synchronization faults is {\yellow identical to} the distribution of the actual p2p synchronization faults, the NCS equation system $\mat{A} \vec{x} = \vec{b}$ has a unique solution that gives correct clock offset estimates. Otherwise, it has no solution. Therefore, if an NCS graph $G$ is $(2K+1)$-edge-connected, it is $K$-resilient, i.e., $C_2 \Rightarrow C_1$.

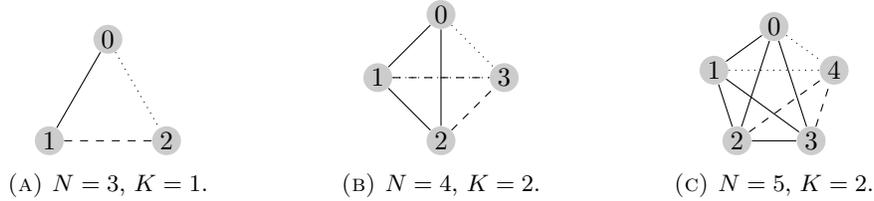
\begin{figure}[b]
	\begin{subfigure}[t]{.30\textwidth}
		\centering
		\def\R{3}
		\begin{tikzpicture}
		[scale=.3,auto=center,every node/.style={circle,fill=black!20,inner sep=1pt}]
		\node (n0) at ({\R*cos(90)},{\R*sin(90)}) {0};
		\node (n1) at ({\R*cos(210)},{\R*sin(210)}) {1};
		\node (n2) at ({\R*cos(330)},{\R*sin(330)}) {2};
		
		\foreach \from/\to in {n0/n1}
		\draw (\from) -- (\to);
		\draw [dotted] (n0) -- (n2);
		\draw [dashed] (n1) -- (n2);
		\end{tikzpicture}
		\centering
		\caption{$N=3$, $K = 1$.}
	\end{subfigure}
	\hfill
	\begin{subfigure}[t]{.30\textwidth}
		\centering
		\def\R{2.8}
		\begin{tikzpicture}
		[scale=.3,auto=center,every node/.style={circle,fill=black!20,inner sep=1pt}]
		\node (n0) at ({\R*cos(90)},{\R*sin(90)}) {0};
		\node (n1) at ({\R*cos(180)},{\R*sin(180)}) {1};
		\node (n2) at ({\R*cos(270)},{\R*sin(270)}) {2};
		\node (n3) at ({\R*cos(360)},{\R*sin(360)}) {3};

		\foreach \from/\to in {n0/n1,n0/n2,n1/n2}
		\draw (\from) -- (\to);
		\foreach \from/\to in {n0/n3,n1/n3}
	    \draw [dotted] (\from) -- (\to);
	    \foreach \from/\to in {n2/n3,n1/n3}
	    \draw [dashed] (\from) -- (\to);
		\end{tikzpicture}
		\centering
		\caption{$N=4$, $K = 2$.}
	\end{subfigure}
	\hfill
	\begin{subfigure}[t]{.30\textwidth}
		\centering
		\def\R{2.8}
		\begin{tikzpicture}
		[scale=.3,auto=center,every node/.style={circle,fill=black!20,inner sep=1pt}]
		\node (n0) at ({\R*cos(90)},{\R*sin(90)}) {0};
		\node (n1) at ({\R*cos(162)},{\R*sin(162)}) {1};
		\node (n2) at ({\R*cos(234)},{\R*sin(234)}) {2};
		\node (n3) at ({\R*cos(306)},{\R*sin(306)}) {3};
		\node (n4) at ({\R*cos(18)},{\R*sin(18)}) {4};
		
		\foreach \from/\to in {n0/n1,n0/n2,n0/n3,n1/n2,n1/n3,n2/n3}
		\draw (\from) -- (\to);
		\foreach \from/\to in {n0/n4,n1/n4}
		\draw [dotted] (\from) -- (\to);
		\foreach \from/\to in {n2/n4,n3/n4}
		\draw [dashed] (\from) -- (\to);
		\end{tikzpicture}
		\centering
		\caption{$N=5$, $K = 2$.}
	\end{subfigure}
	
	\caption{The counterexamples for the small-scale NCS graphs with $N=3,4,5$. Dotted edges denote the edges containing  the actual synchronization faults; dashed edges denote the edges containing the estimated synchronization faults; the dashed-dotted edges denote the edges containing both the estimated and actual synchronization faults. The combination of all the above three types of edges composes of the minimum cut $C$ in a certain NCS graph. The solid edges are fault-free edges. }
	\label{fig:counterexample}
\end{figure}   

    \textbf{Proof of forward implication $C_1 \Rightarrow C_2$.} We have the following equivalence: $(C_1 \Rightarrow C_2) \Leftrightarrow (\neg C_1 \Leftarrow \neg C_2)$.  The $\neg C_2$ means that the NCS graph $G$ is not $(2K+1)$-edge-connected; $\neg C_1$ means that $G$ is not $K$-resilient. In what follows, we prove $\neg C_1 \Leftarrow \neg C_2$. From the definition of $K$-resilience condition in Definition~\ref{def:resilience-condition}, $G$ is not $K$-resilient if we can find any of the following counterexamples: (1) Algorithm~\ref{alg:error-correction} returns a solution when the distribution of the estimated p2p synchronization faults is different from the distribution of the actual faults or (2) Algorithm~\ref{alg:error-correction} returns more than one solution when the distribution of the estimated p2p synchronization faults is identical to the the distribution of the actual faults. In the following, we find such counterexamples when $G$ is not $(2K+1)$-edge-connected. From the Menger's theorem, since $G$ is not $(2K+1)$-edge-connected, there is a minimum edge cut $C$ including at most $2K$ edges for a certain {\yellow pair of nodes $n_i$ and $n_j$}. The minimum edge cut $C$ partitions $G$ into two connected subgraphs $G_i$ and $G_j$ that are disconnected {\yellow from each other}, {\yellow where }$n_i \in \mathcal{V}(G_i)$, $n_j \in \mathcal{V}(G_j)$, and $\mathcal{V}(G)$ represents the set of $G$'s vertexes. Note that the reference node $n_0$ is either in the subgraph $G_i$ or the subgraph $G_j$. Without loss of generality, we assume that $n_0 \in \mathcal{V}(G_i)$. Our counterexamples satisfy the following conditions:  (1) there are $K$ actual faulty sessions and the number of estimated faults is equal to $K$, (2) all the estimated faults and the actual faults are on the edge cut $C$ and each edge of the edge cut $C$  is associated with an estimated fault, or an actual fault, or both of them, (3) all the actual faults have an identical value $e$. {\em For example,
    Fig.~\ref{fig:counterexample} shows the counterexamples for the small-scale networks discussed in \sect\ref{subsec:manual-check}.}

     For the counterexamples described above, the following Eq.~(\ref{eq:wrongsolution}) is a solution of Eq.~(\ref{eq:error-correction}):

\begin{equation}
\left\{
\begin{array}{ll}
\hat{\delta}_{k0}=\delta_{k0}, &\forall n_k \in \mathcal{V}(G_i);\\
\hat{\delta}_{k0}=\delta_{k0}-e,& \forall n_k \in \mathcal{V}(G_j);\\
\hat{e}_{ij}= 0,& \text{if there is an actual fault on} \> n_i \leftrightarrow n_j;\\
\hat{e}_{ij}=-e,& \text{if there is no actual fault on} \> n_i \leftrightarrow n_j.
\end{array}
\right.
\label{eq:wrongsolution}
\end{equation}
\textit{For example, in Fig.~\ref{fig:counterexample}(b), Eq.~(\ref{eq:wrongsolution}) is}
\begin{equation}
	\left\{
	\begin{array}{ll}
		\hat{\delta}_{10}={\delta}_{10};\\
		\hat{\delta}_{20}={\delta}_{20};\\
		\hat{\delta}_{30}={\delta}_{30}-e;\\
	    \hat{e}_{13}=0;\\
	    \hat{e}_{23}=-e.
	\end{array}
	\right.
	\label{eq:examplesoution}
\end{equation}
 Now, we prove that Eq.~(\ref{eq:wrongsolution}) is a solution of Eq.~(\ref{eq:error-correction}) in the counterexamples. We prove it by substituting Eq.~(\ref{eq:wrongsolution}) to Eq.~(\ref{eq:error-correction}). If there is no conflict on each equation of the equation system in Eq.~(\ref{eq:error-correction}) (i.e., Eq.~(\ref{eq:error-correction}) still holds after incorporating Eq.~(\ref{eq:wrongsolution})), Eq.~(\ref{eq:wrongsolution}) is a solution of Eq.~(\ref{eq:error-correction}). Note that Eq.~(\ref{eq:error-correction}) can be separated into three disjoint subequation systems corresponding to the subgraph $G_i$, $G_j$ and the edge cut $C$. Thus, if we substitute Eq.~(\ref{eq:error-correction}) to the three disjoint subequation systems and there is no conflict on each equation, Eq.~(\ref{eq:wrongsolution}) is a solution of Eq.~(\ref{eq:error-correction}). Now, we analyze each of the disjoint subequation systems incorporated with Eq.~(\ref{eq:error-correction}) as follows.

\begin{enumerate}
\item In the subequation system associated with $G_i$, since there are no edges associated with the estimated or actual synchronization fault, the corresponding equation in Eq.~(\ref{eq:error-correction}), denoted by $\mathcal{E}$, has the form of $\hat{\delta}_{w0} - \hat{\delta}_{v0}  = \widetilde{\delta}_{wv}= \delta_{w0}-\delta_{v0}$, where $n_w \in \mathcal{V}(G_i)$ and $n_v \in \mathcal{V}(G_i)$. By substituting Eq.~(\ref{eq:wrongsolution}) to the left-hand side of $\mathcal{E}$, we have $\hat{\delta}_{w0} - \hat{\delta}_{v0} = \delta_{w0}-\delta_{v0}$, which is equal to the right-hand side of $\mathcal{E}$. Thus, there is no conflict when substituting Eq.~(\ref{eq:wrongsolution}) to the subequation system associated with $G_i$. \textit{For example, in Fig.~\ref{fig:counterexample}(b), the subequation system associated with $G_i$ is $\{\hat{\delta}_{10}=\delta_{10}, \hat{\delta}_{20}=\delta_{20},\hat{\delta}_{10}- \hat{\delta}_{20}=\delta_{10}-\delta_{20}\}$. Substituting Eq.~(\ref{eq:examplesoution}) into the above subequation system does not result in any conflict.}
\item In the subequation system associated with $G_j$, since there are no edges associated with the estimated or actual synchronization fault, the corresponding equation $\mathcal{E}$ in Eq.~(\ref{eq:error-correction}) has the form of $\hat{\delta}_{w0} - \hat{\delta}_{v0}  = \widetilde{\delta}_{wv}= \delta_{w0}-\delta_{v0}$, where $n_w \in \mathcal{V}(G_j)$ and $n_v \in \mathcal{V}(G_j)$. By substituting Eq.~(\ref{eq:wrongsolution}) to the left-hand side of $\mathcal{E}$, we have $\hat{\delta}_{w0} - \hat{\delta}_{v0} = (\delta_{w0}-e)-(\delta_{v0}-e) = \delta_{w0}-\delta_{v0}$, which is equal to the right-hand side of $\mathcal{E}$. Thus, there is no conflict when substituting Eq.~(\ref{eq:wrongsolution}) to the subequation system associated with $G_j$. \textit{For example, in Fig.~\ref{fig:counterexample}(b), the subequation system associated with $G_j$ is an empty set $\emptyset$, because there is only one node in the subgraph $G_j$. Thus, there is no conflict between the empty set and the solution in Eq.~(\ref{eq:examplesoution}).}
\item In the subequation system associated with $C$, since each edge is associated with at least one fault between estimated and actual synchronization fault, the corresponding equation $\mathcal{E}$ in Eq.~(\ref{eq:error-correction}) can be in any of the following three forms:
    \begin{enumerate}
    \item If the edge corresponding to $\mathcal{E}$ contains an estimated synchronization fault but no actual synchronization fault, $\mathcal{E}$ has the form of $\hat{\delta}_{w0} - \hat{\delta}_{v0} + \hat{e}_{wv} = \widetilde{\delta}_{wv}= \delta_{w0}-\delta_{v0}$, where $n_w \in \mathcal{V}(G_i)$ and $n_v \in \mathcal{V}(G_j)$. By substituting Eq.~(\ref{eq:wrongsolution}) to the left-hand side of $\mathcal{E}$, we have $\hat{\delta}_{w0} - \hat{\delta}_{v0} +\hat{e}_{wv} = \delta_{w0} - (\delta_{v0}-e) + (-e) =  \delta_{w0}-\delta_{v0}$, which is equal to the right-hand side of $\mathcal{E}$. Thus, there is no conflict when substituting Eq.~(\ref{eq:wrongsolution}) to the subequation system associated with $C$ in this case. \textit{For example, in Fig.~\ref{fig:counterexample}(b), the subequation system associated with $C$ in this case is $\hat{\delta}_{03}= \delta_{03} + e$, which is also in Eq.~(\ref{eq:examplesoution}). Thus, there is no conflict.}
    \item If the edge corresponding to $\mathcal{E}$ contains an actual synchronization fault but no estimated synchronization fault, $\mathcal{E}$ has the form of $\hat{\delta}_{w0} - \hat{\delta}_{v0}  = \widetilde{\delta}_{wv}= \delta_{w0}-\delta_{v0} + e_{wv}$, where $n_w \in \mathcal{V}(G_i)$ and $n_v \in \mathcal{V}(G_j)$. By substituting Eq.~(\ref{eq:wrongsolution}) to the left-hand side of $\mathcal{E}$, we have $\hat{\delta}_{w0} - \hat{\delta}_{v0} = \delta_{w0} - (\delta_{v0}-e) =  \delta_{w0}-\delta_{v0} +e$. Note that in our counterexamples, all the actual faults have an identical value $e$, i.e., $e_{wv} =e$. Thus, $\delta_{w0}-\delta_{v0} +e$ is equal to the left-hand side of $\mathcal{E}$. There is no conflict when substituting Eq.~(\ref{eq:wrongsolution}) to the subequation system associated with $C$ in this case. \textit{For example, in Fig.~\ref{fig:counterexample}(b), the subequation system associated with $C$ in this case is $\hat{\delta}_{20}-\hat{\delta}_{30} + \hat{e}_{23}= \delta_{30} - \delta_{20}$. Note that in Eq.~(\ref{eq:examplesoution}), $\hat{e}_{23} = -e$. Therefore, if we substitute Eq.~(\ref{eq:examplesoution}) into the above subequation system, there is no conflict.} 
    \item If the edge corresponding to $\mathcal{E}$ contains both an actual synchronization fault and an estimated synchronization fault, $\mathcal{E}$ has the form of $\hat{\delta}_{w0} - \hat{\delta}_{v0} +\hat{e}_{wv} = \widetilde{\delta}_{wv}= \delta_{w0}-\delta_{v0} + e_{wv}$, where $n_w \in \mathcal{V}(G_i)$ and $n_v \in \mathcal{V}(G_j)$. Note that in this case $\hat{e}_{wv} = 0$ in Eq.~(\ref{eq:wrongsolution}). By substituting Eq.~(\ref{eq:wrongsolution}) to the left-hand side of $\mathcal{E}$, we have $\hat{\delta}_{w0} - \hat{\delta}_{v0} +\hat{e}_{wv}= \delta_{w0} - (\delta_{v0}-e) + 0 =  \delta_{w0}-\delta_{v0} +e$, which is equal to the right-hand side of $\mathcal{E}$. Thus, there is no conflict when substituting Eq.~(\ref{eq:wrongsolution}) to the subequation system associated with $C$ in this case. \textit{For example, in Fig.~\ref{fig:counterexample}(b), the subequation system associated with $C$ in this case is $\hat{\delta}_{10}-\hat{\delta}_{30} + \hat{e}_{13}= \delta_{10} - \delta_{30} + e$. Note that in Eq.~(\ref{eq:examplesoution}), $\hat{e}_{13} = 0$. Therefore, substituting Eq.~(\ref{eq:examplesoution}) into the above subequation system results in no conflict.}
    \end{enumerate}
    Recall that in our counterexamples, each edge of the edge cut $C$ is associated with an estimated fault, or an actual fault, or both of them. Therefore, there is no conflict when substituting Eq.~(\ref{eq:wrongsolution}) to the subequation system associated with $C$.
\end{enumerate}

In summary, we have substituted Eq.~(\ref{eq:wrongsolution}) to the three disjoint subequation systems and there is no conflict. Thus, Eq.~(\ref{eq:wrongsolution}) is a solution of Eq.~(\ref{eq:error-correction}). Now, we prove that in this case, $G$ is not $K$-resilient. If the distribution of the estimated p2p faults is identical to the distribution of the actual faults, Eq.~(\ref{eq:error-correction}) has at least two solutions including Eq.~(\ref{eq:wrongsolution}) and the solution $\{\hat{\vecg{\delta}}=\vecg{\delta} , \hat{\vecg{e}}=\vecg{e}\}$, which violates the condition  (2)-(a) of the $K$-resilience condition defined in Definition~\ref{def:resilience-condition}. If the two distributions are different, Eq.~(\ref{eq:error-correction}) has a solution of Eq.~(\ref{eq:wrongsolution}), which violates the condition (2)-(b) of the $K$-resilience condition defined in Definition~\ref{def:resilience-condition}. Therefore, for any NCS graph $G$ that is not ($2K+1$)-edge-connected, $G$ is not $K$-resilient, i.e., $\neg C_1 \Leftarrow \neg C_2$. Therefore, $C_1 \Rightarrow C_2$.
\end{proof}

\subsection{Algorithm to Compute Tight Bound of Maximum Resilience}
\label{subsec:alguncomplete}

Based on Theorem~\ref{thm:new1}, Algorithm~\ref{alg:tightboundnew} computes the tight bound of the maximum resilience for {\yellow NCS graph $G$}. Specifically, starting with $K=0$, {\yellow Algorithm~\ref{alg:tightboundnew}} increases $K$ by one in each step of the outer loop to check whether {\yellow the NCS graph} is $K$-resilient by checking {\yellow the connectivity of the subgraphs after removing $2K$ edges from $G$}. If {\yellow any subgraph} is not connected, the sufficient and necessary condition given by Theorem~\ref{thm:new1} is not satisfied for the current $K$ value. Thus, the algorithm returns $K-1$ as the tight bound.

\begin{algorithm}[t]
  \caption{Compute the tight bound of maximum resilience}
\label{alg:tightboundnew}
\begin{algorithmic}[1]
\REQUIRE NCS equation system $\mat{A}\vec{x}=\vec{b}$ and the corresponding NCS graph $G = (V, E)$
\ENSURE Tight bound of maximum resilience of $G$

\STATE $K \leftarrow 0$
\WHILE{$K \le |E|$}
\FOR{each combination of $2K$ edges selected from all the $|E|$ edges in $G$}
\STATE remove the selected $2K$ edges to generate a subgraph $G'$
\IF{$G'$ is not connected}
\RETURN{$K-1$}
\ENDIF
\ENDFOR
\STATE $K \leftarrow K + 1$
\ENDWHILE
\end{algorithmic}
\end{algorithm}

{\yellow Now, we analyze the time complexity of Algorithm~\ref{alg:tightboundnew}. For each $K$ value, Algorithm~\ref{alg:tightboundnew} needs to check the connectivity of totally ${|E| \choose 2K}$ subgraphs. Existing  graph-theoretic algorithms can be used to check the connectivity of a graph, such as depth-first search (DFS) and breadth-first search (BFS) \cite{cormen2009introduction}. The DFS and BFS algorithms have the same time complexity of $O(|V|+|E|)$. In particular, for complete NCS graphs, the complexity of the two algorithms is $O(|V|^2)$.
Thus, the time complexity of the $K$th step of Algorithm~\ref{alg:tightboundnew} is $O\left(|V|^2{\frac{|V||V-1|}{2} \choose 2K}\right)$. Therefore, determining the tight bound of maximum resilience for any graph incurs a high computation overhead for large-scale NCS graphs. Nevertheless, Algorithm~\ref{alg:tightboundnew} is a method to exactly compute the tight bound of maximum resilience for incomplete NCS graphs. Figure~\ref{fig:uncomplete} shows several incomplete NCS graphs and their tight bounds of maximum resilience computed by Algorithm~\ref{alg:tightboundnew}.
 }

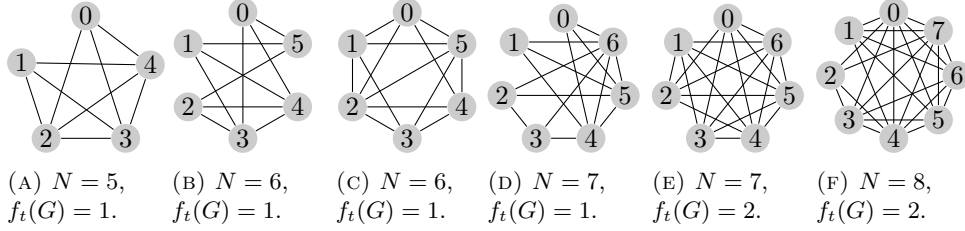
\begin{figure}
	\begin{subfigure}[t]{.15\textwidth}
		\centering
		\def\R{3}
		\begin{tikzpicture}
		[scale=.3,auto=center,every node/.style={circle,fill=black!20,inner sep=1pt}]
		\node (n0) at ({\R*cos(90)},{\R*sin(90)}) {0};
		\node (n1) at ({\R*cos(162)},{\R*sin(162)}) {1};
		\node (n2) at ({\R*cos(234)},{\R*sin(234)}) {2};
		\node (n3) at ({\R*cos(306)},{\R*sin(306)}) {3};
		\node (n4) at ({\R*cos(18)},{\R*sin(16)}) {4};
		
		\foreach \from/\to in {n0/n2,n0/n3,n0/n4,n1/n2,n1/n3,n1/n4,n2/n4,n3/n4,n2/n3}
		\draw (\from) -- (\to);
		\end{tikzpicture}
        \centering
		\caption{$N=5$, \\$f_t(G) = 1$.}
	\end{subfigure}
	\hfill
	\begin{subfigure}[t]{.15\textwidth}
		\centering
		\def\R{2.8}
		\begin{tikzpicture}
		[scale=.3,auto=center,every node/.style={circle,fill=black!20,inner sep=1pt}]
		\node (n0) at ({\R*cos(90)},{\R*sin(90)}) {0};
		\node (n1) at ({\R*cos(150)},{\R*sin(150)}) {1};
		\node (n2) at ({\R*cos(210)},{\R*sin(210)}) {2};
		\node (n3) at ({\R*cos(270)},{\R*sin(270)}) {3};
		\node (n4) at ({\R*cos(330)},{\R*sin(330)}) {4};
		\node (n5) at ({\R*cos(30)},{\R*sin(30)}) {5};
		
		\foreach \from/\to in {n0/n3,n0/n4,n0/n5,n1/n3,n1/n4,n1/n5,n2/n3,n2/n4,n2/n5,n3/n4}
		\draw (\from) -- (\to);
		\end{tikzpicture}
        \centering
		\caption{$N=6$, \\$f_t(G) = 1$.}
	\end{subfigure}
	\hfill
	\begin{subfigure}[t]{.15\textwidth}
		\centering
		\def\R{2.8}
		\begin{tikzpicture}
		[scale=.3,auto=center,every node/.style={circle,fill=black!20,inner sep=1pt}]
		\node (n0) at ({\R*cos(90)},{\R*sin(90)}) {0};
		\node (n1) at ({\R*cos(150)},{\R*sin(150)}) {1};
		\node (n2) at ({\R*cos(210)},{\R*sin(210)}) {2};
		\node (n3) at ({\R*cos(270)},{\R*sin(270)}) {3};
		\node (n4) at ({\R*cos(330)},{\R*sin(330)}) {4};
		\node (n5) at ({\R*cos(30)},{\R*sin(30)}) {5};
		
		\foreach \from/\to in {n0/n2,n0/n1,n0/n4,n0/n5,n1/n2,n1/n3,n1/n5,n2/n3,n2/n5,n2/n4,n3/n4,n3/n5,n4/n5}
		\draw (\from) -- (\to);
		\end{tikzpicture}
        \centering
		\caption{$N=6$, \\$f_t(G) = 1$.}
	\end{subfigure}
	\vspace{1em}
	\begin{subfigure}[t]{.15\textwidth}
		\centering
		\def\R{2.8}
		\begin{tikzpicture}
		[scale=.3,auto=center,every node/.style={circle,fill=black!20,inner sep=1pt}]
		\node (n0) at ({\R*cos(90)},{\R*sin(90)}) {0};
		\node (n1) at ({\R*cos(141.428571429)},{\R*sin(141.428571429)}) {1};
		\node (n2) at ({\R*cos(192.857142857)},{\R*sin(192.857142857)}) {2};
		\node (n3) at ({\R*cos(244.285714286)},{\R*sin(244.285714286)}) {3};
		\node (n4) at ({\R*cos(295.714285714)},{\R*sin(295.714285714)}) {4};
		\node (n5) at ({\R*cos(347.142857143)},{\R*sin(347.142857143)}) {5};
		\node (n6) at ({\R*cos(38.5714285714)},{\R*sin(38.5714285714)}) {6};
		
		\foreach \from/\to in {n0/n4,n0/n5,n0/n6,n1/n4,n1/n5,n1/n6,n2/n3,n2/n5,n2/n6,n3/n4,n3/n6,n4/n5,n4/n6,n5/n6}
		\draw (\from) -- (\to);
		\end{tikzpicture}
        \centering
		\caption{$N=7$, \\$f_t(G) = 1$.}
	\end{subfigure}
	\hfill
	\begin{subfigure}[t]{.15\textwidth}
		\centering
		\def\R{2.8}
		\begin{tikzpicture}
		[scale=.3,auto=center,every node/.style={circle,fill=black!20,inner sep=1pt}]
		\node (n0) at ({\R*cos(90)},{\R*sin(90)}) {0};
		\node (n1) at ({\R*cos(141.428571429)},{\R*sin(141.428571429)}) {1};
		\node (n2) at ({\R*cos(192.857142857)},{\R*sin(192.857142857)}) {2};
		\node (n3) at ({\R*cos(244.285714286)},{\R*sin(244.285714286)}) {3};
		\node (n4) at ({\R*cos(295.714285714)},{\R*sin(295.714285714)}) {4};
		\node (n5) at ({\R*cos(347.142857143)},{\R*sin(347.142857143)}) {5};
		\node (n6) at ({\R*cos(38.5714285714)},{\R*sin(38.5714285714)}) {6};
		
		\foreach \from/\to in {n0/n2,n0/n3,n0/n4,n0/n5,n0/n6,n1/n2,n1/n3,n1/n4,n1/n5,n1/n6,n2/n4,n2/n5,n2/n6,n3/n4,n3/n5,n3/n6,n4/n6,n5/n6,n4/n5}
		\draw (\from) -- (\to);
		\end{tikzpicture}
        \centering
		\caption{$N=7$, \\$f_t(G)=2$.}
	\end{subfigure}
	\hfill
	\begin{subfigure}[t]{.15\textwidth}
		\centering
		\def\R{2.8}
		\begin{tikzpicture}
		[scale=.3,auto=center,every node/.style={circle,fill=black!20,inner sep=1pt}]
		\node (n0) at ({\R*cos(90)},{\R*sin(90)}) {0};
		\node (n1) at ({\R*cos(135)},{\R*sin(135)}) {1};
		\node (n2) at ({\R*cos(180)},{\R*sin(180)}) {2};
		\node (n3) at ({\R*cos(225)},{\R*sin(225)}) {3};
		\node (n4) at ({\R*cos(270)},{\R*sin(270)}) {4};
		\node (n5) at ({\R*cos(315)},{\R*sin(315)}) {5};
		\node (n6) at ({\R*cos(0)},{\R*sin(0)}) {6};
		\node (n7) at ({\R*cos(45)},{\R*sin(45)}) {7};
		
		\foreach \from/\to in {n0/n1,n0/n2,n0/n3,n0/n4,n0/n5,n0/n6,n0/n7,n1/n4,n1/n5,n1/n6,n1/n7,n2/n3,n2/n4,n2/n7,n3/n4,n3/n5,n3/n6,n3/n7,n4/n5,n4/n6,n4/n7,n5/n6,n5/n7,n6/n7}
		\draw (\from) -- (\to);
		\end{tikzpicture}
        \centering
		\caption{$N=8$, \\$f_t(G) = 2$.}
	\end{subfigure}
	\caption{The tight bounds of maximum resilience for several incomplete NCS graphs with $N=5,6,7,8$.}
	\label{fig:uncomplete}
      \end{figure}

\section{Fast NCS Algorithm with Fault Correction}
\label{sec:newal}
Algorithm~\ref{alg:error-correction} enumerates all possible distributions of the faults, leading to the exponential time complexity in the worst case. From the proof of Theorem~\ref{thm:new1}, if we can find a fault-free path connecting $n_i$ and $n_0$, we can obtain the correct estimate of clock offset between $n_i$ and $n_0$. This  observation sheds light on a new NCS algorithm that can correct the faults without enumerating all possible distributions of the faults. In this section, we present such a new NCS algorithm. Then, we show that the new NCS algorithm achieves the same fault correction capability as Algorithm~\ref{alg:error-correction} and analyze the time complexity of the new NCS algorithm.

\subsection{Fast NCS Algorithm}
\label{subsec:newalgorithm}
First, we prove that we can correctly estimate the p2p clock synchronization offsets on the fault-free path. We have the following lemma.
\begin{lemma}
\label{lem:newlemfault}
Any fault-free path between the reference node $n_0$ and node $n_i$ {\yellow leads to the correct estimate of $n_i$'s clock offset, i.e.}, $\hat{\delta_{i}} = \delta_{i}$.
\end{lemma}
\begin{proof}
 {\yellow Denote the path by} $\langle n_0, n_{w_1}, n_{w_2}, \ldots, n_{w_p}, n_i \rangle$. We can formulate a system of equations along the path. Since all edges on the path are associated with neither estimated nor actual synchronization fault, the equation system consists of a chain $\hat{\delta}_{w_10}=\delta_{w_10}$, $\hat{\delta}_{w_10}-\hat{\delta}_{w_20}=\delta_{w_10}-\delta_{w_20}$, $\hat{\delta}_{w_20}-\hat{\delta}_{w_30}=\delta_{w_20}-\delta_{w_30}$, $\ldots$, $\hat{\delta}_{w_{p-1}0}-\hat{\delta}_{w_p0}=\delta_{w_{p-1}0}-\delta_{w_p0}$, $\hat{\delta}_{w_p0}-\hat{\delta}_{w_i0}=\delta_{w_p0}-\delta_{w_i0}$. By substituting the solution of the previous equation to the next equation in the above chain of equations, we have a solution that $\hat{\delta_{i}} = \delta_{i}$.
\end{proof}
{\yellow From Lemma~3, if} for every node $n_i$ in $G$ we can find at least one fault-free path connecting $n_i$ and $n_0$, we can obtain {\yellow all the correct clock offset estimates, i.e.}, $\hat{\vecg{\delta}}=\vecg{\delta}$. Then, we can use the solution $\hat{\vecg{\delta}}=\vecg{\delta}$ to {\yellow pinpoint the faulty p2p synchronization sessions. Specifically, if} {\yellow $\widetilde{\delta}_{ij} \neq \hat{\delta}_i - \hat{\delta}_j$, where $\hat{\delta}_i$ and $\hat{\delta}_j$ are from $\hat{\vecg{\delta}}$}, the p2p synchronization session between $n_i$ and $n_j$ is faulty. The fault is {\yellow given by} $e_{ij}= \widetilde{\delta}_{ij}-(\hat{\delta_{i}}-\hat{\delta_{j}})$.

 {\yellow Thus, the NCS problem becomes} how to find a fault-free path between any node $n_i$ and the reference node $n_0$. {\yellow This is challenging because the system has no knowledge of the number of faults and their distribution among the $|E|$ sessions.} We {\yellow address this challenge using a voting scheme}. The details are as follows. {\yellow We let $Z_i$ denote the maximum number of pairwise edge-disjoint paths connecting $n_i$ and $n_0$ and let $S_i$ denote a set of such paths. Existing algorithms can be used to compute $Z_i$ and $S_i$, such as those presented in \cite{kaufmann1991faster}, \cite{wagner1995linear} and \cite{perl1978finding}. The worst-case time complexity of these algorithms is $O(|V|^2)$. Assuming every path in $S_i$ is fault-free, we apply the approach described in Lemma~\ref{lem:newlemfault} to compute $\hat{\delta}_i$ for every path in $S_i$.} Note that {\yellow there might be multiple different sets of pairwise edge-disjoint paths with the identical set cardinality. The $S_i$ used in the following discussion can be any one of them.} If all the paths in the set $S_i$ are really fault-free, {\yellow their corresponding estimated clock offsets should be the same.} Otherwise, {\yellow they will be different.} We use the most {\yellow frequent} value {\yellow among all the clock offset estimates} as {\yellow the voting result, which is yielded as the final clock offset estimate $\hat{\delta}_i$.}
 {\yellow We repeat the above process for every node $n_i$ to generate the voting result $\hat{\delta}_i$. After that, we can correct the faults by following the procedure described in last paragraph. Algorithm~\ref{alg:error-correction2} shows the pseudocode of the new NCS algorithm. Algorithm~\ref{alg:error-correction2} has the same practicality as Algorithm~\ref{alg:error-correction} in that it requires neither the actual number nor the actual distribution of the p2p synchronization faults.}

\begin{algorithm}[t]
\caption{{\blue Fast NCS algorithm with fault correction}.}
\label{alg:error-correction2}
\begin{algorithmic}[1]
\REQUIRE $\{ \widetilde{\delta}_{ij} | \forall n_i \leftrightarrow n_j \in E \}$
\ENSURE $\{ \hat{\delta}_{j0}| \forall j \in [1, N-1] \}$ and $\{\hat{e}_{ij} | \forall n_i \leftrightarrow n_j \in E \}$
\label{line:start2}
\FOR{each node $n_i$ in $V$  where $i \neq 0$}
\label{line:foreach2}
\STATE compute the maximum number of pairwise edge-disjoint paths $Z_i$ and find a corresponding set of such paths $S_i$
\FOR{each path $P_k \in S_i$ }
\STATE compute the corresponding value of the estimated clock offset $\hat{\delta}_{i0}^k$
\ENDFOR
\STATE $\hat{\delta}_{i0}$ $\leftarrow$ the most frequent value in ${\hat\delta^{k}_{i0}}$, where $k \in \{1,2,3,...,Z_i\}$
\label{line:judge}
\ENDFOR
\label{line:endfor2}
\FOR{each synchronization session $n_i \leftrightarrow n_j$}
\IF{$\widetilde{\delta}_{ij}-(\hat{\delta_{i}}-\hat{\delta_{j}}) \neq 0$}
\STATE $\hat{e}_{ij}=\widetilde{\delta}_{ij}-(\hat{\delta_{i}}-\hat{\delta_{j}})$
\ENDIF
\ENDFOR
\RETURN $\{ \hat{\delta}_{j0}| \forall j \in [1, N-1] \}$ and $\{\hat{e}_{ij} | \forall n_i \leftrightarrow n_j \in E \}$
\end{algorithmic}
\end{algorithm}

\subsection{Tight Bound of Maximum Resilience of Fast NCS with Fault Correction}
\label{subsec:comparision}
{\yellow In this section, we show that Algorithms~\ref{alg:error-correction} and~\ref{alg:error-correction2} have the same fault correction capability. Therefore, the networks with Algorithms~\ref{alg:error-correction} and~\ref{alg:error-correction2} as the NCS algorithm respectively have the same tight bound of maximum resilience.}
\begin{theorem}
\label{thm:same}
 {\yellow For any NCS graph $G$, Algorithms~\ref{alg:error-correction} and~\ref{alg:error-correction2} achieve the same tight bound of maximum resilience.}
\end{theorem}
\begin{proof}
 {\yellow First, we prove that, if an NCS graph is $K$-resilient under Algorithm~\ref{alg:error-correction}, it is also $K$-resilient under Algorithm~3.} {\yellow From Theorem~\ref{thm:new1}, the $K$-resilience of $G$ under Algorithm~\ref{alg:error-correction} is equivalent to that} $G$ is $(2K+1)$-edge-connected. From the Menger's theorem{\yellow, the $(2K+1)$-edge-connectivity means that} for any {\yellow pair of nodes} $n_i$ and $n_j$, there exist at least $2K+1$ edge-disjoint paths connecting them. Since the number of actual faults is no greater than $K$, there are at least $K+1$ fault-free paths between $n_i$ and $n_0$. {\yellow Thus, the majority voting in Algorithm~\ref{alg:error-correction2} must give the correct result and Algorithm~\ref{alg:error-correction2} can correct the faults. Therefore, the NCS graph is also $K$-resilient under Algorithm~\ref{alg:error-correction2}.}

{\yellow Then, we prove that, if an NCS graph $G$ is not $K'$-resilient under Algorithm~\ref{alg:error-correction}, it is also not $K'$-resilient under Algorithm~\ref{alg:error-correction2}. We assume $G$ with Algorithm~\ref{alg:error-correction} can correct at most $K$ faults, where $K' \ge K +1$. From Theorem~\ref{thm:new1}, $G$ is $(2K+1)$-edge-connected.} Thus, there exists at least one node pair $n_i$ and $n_j$ that have $2K+1$ edge-disjoint paths and no more connecting them. {\yellow Note that the system's resilience is independent from the choice of reference node, i.e., any node can be designated as the reference node $n_0$. Without loss of generality, we designate $n_j$ as $n_0$. Now, we consider the cases where there are $K'$ faults, i.e., there are at least $K+1$ faults since $K' \ge K+1$. For the case where all the $K+1$ faults occur on the paths among the $2K+1$ edge-disjoint paths, the remaining fault-free edge-disjoint paths do not form the majority of Algorithm~3's voting. As a result, Algorithm~\ref{alg:error-correction2} cannot correctly estimate the clock offset of $n_i$ and cannot correct the faults. Thus, the $G$ with Algorithm~{\ref{alg:error-correction2}} is not $K'$-resilient.}
\end{proof}

{\yellow  Now, we analyze the time complexity of Algorithms~\ref{alg:error-correction} and~\ref{alg:error-correction2}. If the algorithm in \cite{perl1978finding} is used to compute $Z_i$ and $S_1$, Line~2 of Algorithm~\ref{alg:error-correction2} has a time complexity of $O(|V|^2)$. The loop from Line~3 to Line~5 has a time complexity of $O(|V|)$, because from Theorems~\ref{thm:1} and~\ref{thm:new1}, the maximum number of pairwise edge-disjoint paths is less than $ \left\lfloor \frac{|V|}{2} \right\rfloor - 1$. Line~6 has a time complexity of $O(|V|)$ \cite{parhami1994voting}. Thus, the loop from Line~1 to Line~7 has a time complexity of $O(|V|^3)$. The loop from Line~8 to Line~12 has a time complexity $O(|V|^2)$. Therefore, the time complexity of Algorithm~\ref{alg:error-correction2} is $O(|V|^3+|V|^2)=O(|V|^3)$.

In \sect\ref{subsec:system-model}, we have shown that the time complexity upper bound of Algorithm~1 is $O(2^{|E|})$. Now, we derive the time complexity lower bound of Algorithm~1 for complete NCS graphs that are $K$-resilient. When there are $K$ faults, the time complexity of Algorithm~1 is $O \left(  \sum_{k=0}^{K} {|E| \choose k}  \right)$). From Theorem~\ref{thm:1}, Algorithm~\ref{alg:error-correction} can correct at most $K=\left\lfloor \frac{N}{2} \right\rfloor - 1$ faults. Thus, the time complexity of Algorithm~\ref{alg:error-correction} is $O \left(  \sum_{k=0}^{\left\lfloor \frac{N}{2} \right\rfloor - 1} {\frac{N(N-1)}{2} \choose k } \right)$.
When $N > 4$, we have the following inequality:
\begin{equation}
      \sum_{k=0}^{\left\lfloor \frac{N}{2} \right\rfloor - 1} {\frac{N(N-1)}{2} \choose k } >  \sum_{k=0}^{N/2} {N/2 \choose k} = \sqrt{2}^N.
\end{equation}
Thus, Algorithm~\ref{alg:error-correction} has an exponential complexity.
}

From the above analysis, Algorithm~\ref{alg:error-correction2} achieves the same fault correction capability as Algorithm~\ref{alg:error-correction} with a cubic time complexity. In practice, Algorithm~\ref{alg:error-correction2} should be used. Note that as Algorithm~\ref{alg:error-correction} is intuitive, the definition of fault resilience based on Algorithm~\ref{alg:error-correction} is also intuitive. Differently, the development of Algorithm~\ref{alg:error-correction2} is based on our further analysis on the fault resilience. As a result, the fault resilience notion behind Algorithm~\ref{alg:error-correction2} is not direct. From this sense, Algorithm~\ref{alg:error-correction}, though not scalable to the network size, helps achieve a clear definition of fault resilience and is still a basis of this paper.

\section{Minimum NCS Graph for $K$-Resilience}
\label{sec:another}

\sect\ref{sec:fully} and \sect\ref{sec:lower-bounds} analyzed the bounds of the number of faults that Algorithm~\ref{alg:error-correction} and Algorithm~\ref{alg:error-correction2} can correct. Differently, in this section, we aim at minimizing the number of p2p synchronization sessions while maintaining the $K$-resilience of a network under Algorithms~\ref{alg:error-correction} and~\ref{alg:error-correction2}. In other words, we aim at looking for the {\em minimum NCS graph} for $K$-resilience, which is formally defined as follows.
\begin{definition}[Minimum NCS graph for $K$-resilience]
  Denote by $V$ a set of $N$ nodes. An NCS graph $G = (V, E)$ is a minimum NCS graph for $K$-resilience if the network with $G$ is $K$-resilient and any network with the NCS graph $G'=(V, E')$ where $|E'| < |E|$ is not $K$-resilient.
\end{definition}

With minimum NCS graphs, we can minimize the communication cost of NCS without compromising fault correction capability. In \sect\ref{subsec:minimum-ncs}, we develop an algorithm based on the $K$-resilience's sufficient and necessary condition given by Theorem~\ref{thm:new1} to compute the minimum NCS subgraphs and show several examples. In \sect\ref{subsec:min-ncs-condition}, we derive the theoretic lower bound of the number of edges in the NCS graph that provides $K$-resilience. The theoretic lower bound can be used to understand the order of magnitude of the number of edges in a computed minimum NCS graph. In particular, the number of edges of a computed minimum NCS graph is identical to the theoretic lower bound  in our computed examples. {\yellow  This implies that the theoretic lower bound is tight.}

\subsection{The Algorithm to Compute Minimum NCS Graphs}
\label{subsec:minimum-ncs}

\begin{algorithm}[t]
  \caption{Compute minimum NCS graphs for any $N$-node network with $K$-resilience}
\label{alg:quasi-minimum}
\begin{algorithmic}[1]
\REQUIRE The number of nodes $N$, resilience value $K$
\ENSURE  A set of minimum NCS graphs $\mathcal{G}$
\STATE $m \leftarrow 0$, $\mathcal{G} \leftarrow \emptyset$
\STATE construct the $N$-node complete NCS graph $G_c = (V, E_c)$
\WHILE{$m \le \frac{N(N-1)}{2}$}
\label{line:while-loop-start}
\STATE $\mathcal{G}_{\text{current}} \leftarrow \emptyset$
\label{line:reset}
\FOR{each combination of $m$ sessions among $E_c$}
\STATE remove the $m$ sessions from $G_c$ to generate an NCS subgraph $G' =(V,E')$\label{line:min-graph-generate}
\STATE \textit{resilient} $\leftarrow$ \TRUE
\label{line:test-start}
\FOR{each combination of $2K$ equations selected from all the $|E'|$ sessions}
\STATE remove the selected $2K$ edges to generate an NCS subgraph $G''$
\IF{$G''$ is unconnected}
\STATE \textit{resilient} $\leftarrow$ \FALSE $\quad$ \COMMENT{$G'$ is not $K$-resilient}
\STATE break
\ENDIF
\ENDFOR
\label{line:test-end}
\IF{\textit{resilient} = \TRUE}
\STATE $\mathcal{G}_{\text{current}} \leftarrow \mathcal{G}_{\text{current}} \cup \{G'\}$
\label{line:inclusion}
\ENDIF
\ENDFOR
\IF{$\mathcal{G}_{\text{current}} = \emptyset$}
\label{line:return-start}
\RETURN{$\mathcal{G}$} \COMMENT{each $G'$ is not $K$-resilient for current $m$}
\ELSE
\STATE $\mathcal{G} \leftarrow \mathcal{G}_{\text{current}}$
\ENDIF
\STATE $m \leftarrow m+1$
\label{line:return-end}
\ENDWHILE
\label{line:while-loop-end}
\end{algorithmic}
\end{algorithm}

Based on Theorem~\ref{thm:new1}, Algorithm~\ref{alg:quasi-minimum} finds the minimum NCS graphs for any $N$-node network to ensure $K$-resilience. Note that a network may have multiple different minimum NCS graphs. Algorithm~\ref{alg:quasi-minimum} returns a set of minimum NCS graphs that have the same number of edges. We now explain Algorithm~\ref{alg:quasi-minimum}. The algorithm uses an $N$-node complete NCS graph $G_c = (V, E_c)$ as the basis to look for the minimum NCS graphs. In each iteration of the while loop (from Line~\ref{line:while-loop-start} to Line~\ref{line:while-loop-end}), the $m$ is increased by one from zero, where the $m$ represents the number of edges removed from the $E_c$ of $G_c$ to generate a candidate minimum NCS graph $G'$ (Line~\ref{line:min-graph-generate}). For each possible $G'$, the snippet from Line~\ref{line:test-start} to Line~\ref{line:test-end} uses Theorem~\ref{thm:new1} to check whether $G'$ is $K$-resilient. If $G'$ is $K$-resilient, $G'$ is included into a set $\mathcal{G}_{\text{current}}$ (Line~\ref{line:inclusion}) that is reset to an empty set for the next $m$ value (Line~\ref{line:reset}). If all possible $G'$ graphs with the current $m$ value cannot be confirmed $K$-resilient (i.e., $\mathcal{G}_{\text{current}}=\emptyset$), the algorithm returns the non-empty $\mathcal{G}_{\text{current}}$ in the previous iteration of the while loop. This mechanism is implemented by Line~\ref{line:return-start} to Line~\ref{line:return-end}. The algorithm gives the minimum NCS graphs, because the snippet from Line~\ref{line:test-start} to Line~\ref{line:test-end} can confirm the $K$-resilience of the candidate $G'$.

Fig.~\ref{fig:quasi} shows several minimum NCS graphs computed by Algorithm~\ref{alg:quasi-minimum} under different settings of $N$ and $K$. In \sect\ref{subsec:min-ncs-condition}, we will show that the number of edges in all these graphs are equal to the theoretic lower bound.

\subsection{Lower Bound of the Number of Edges for $K$-Resilience}
\label{subsec:min-ncs-condition}
\begin{figure}
	\begin{subfigure}[t]{.15\textwidth}
		\centering
		\def\R{3}
		\begin{tikzpicture}
		[scale=.3,auto=center,every node/.style={circle,fill=black!20,inner sep=1pt}]
		\node (n0) at ({\R*cos(90)},{\R*sin(90)}) {0};
		\node (n1) at ({\R*cos(162)},{\R*sin(162)}) {1};
		\node (n2) at ({\R*cos(234)},{\R*sin(234)}) {2};
		\node (n3) at ({\R*cos(306)},{\R*sin(306)}) {3};
		\node (n4) at ({\R*cos(18)},{\R*sin(16)}) {4};
		
		\foreach \from/\to in {n0/n2,n0/n3,n0/n4,n1/n2,n1/n3,n1/n4,n2/n4,n3/n4}
		\draw (\from) -- (\to);
		\end{tikzpicture}
		\caption{$N=5$, $K=1$.}
	\end{subfigure}
	\hfill
	\begin{subfigure}[t]{.15\textwidth}
		\centering
		\def\R{2.8}
		\begin{tikzpicture}
		[scale=.3,auto=center,every node/.style={circle,fill=black!20,inner sep=1pt}]
		\node (n0) at ({\R*cos(90)},{\R*sin(90)}) {0};
		\node (n1) at ({\R*cos(150)},{\R*sin(150)}) {1};
		\node (n2) at ({\R*cos(210)},{\R*sin(210)}) {2};
		\node (n3) at ({\R*cos(270)},{\R*sin(270)}) {3};
		\node (n4) at ({\R*cos(330)},{\R*sin(330)}) {4};
		\node (n5) at ({\R*cos(30)},{\R*sin(30)}) {5};
		
		\foreach \from/\to in {n0/n3,n0/n4,n0/n5,n1/n3,n1/n4,n1/n5,n2/n3,n2/n4,n2/n5}
		\draw (\from) -- (\to);
		\end{tikzpicture}
		\caption{$N=6$, $K=1$.}
	\end{subfigure}
	\hfill
	\begin{subfigure}[t]{.15\textwidth}
		\centering
		\def\R{2.8}
		\begin{tikzpicture}
		[scale=.3,auto=center,every node/.style={circle,fill=black!20,inner sep=1pt}]
		\node (n0) at ({\R*cos(90)},{\R*sin(90)}) {0};
		\node (n1) at ({\R*cos(150)},{\R*sin(150)}) {1};
		\node (n2) at ({\R*cos(210)},{\R*sin(210)}) {2};
		\node (n3) at ({\R*cos(270)},{\R*sin(270)}) {3};
		\node (n4) at ({\R*cos(330)},{\R*sin(330)}) {4};
		\node (n5) at ({\R*cos(30)},{\R*sin(30)}) {5};
		
		\foreach \from/\to in {n0/n2,n0/n1,n0/n3,n0/n4,n0/n5,n1/n2,n1/n3,n1/n4,n1/n5,n2/n3,n2/n5,n2/n4,n3/n4,n3/n5,n4/n5}
		\draw (\from) -- (\to);
		\end{tikzpicture}
		\caption{$N=6$, $K=2$.}
	\end{subfigure}
	\vspace{1em}
	\begin{subfigure}[t]{.15\textwidth}
		\centering
		\def\R{2.8}
		\begin{tikzpicture}
		[scale=.3,auto=center,every node/.style={circle,fill=black!20,inner sep=1pt}]
		\node (n0) at ({\R*cos(90)},{\R*sin(90)}) {0};
		\node (n1) at ({\R*cos(141.428571429)},{\R*sin(141.428571429)}) {1};
		\node (n2) at ({\R*cos(192.857142857)},{\R*sin(192.857142857)}) {2};
		\node (n3) at ({\R*cos(244.285714286)},{\R*sin(244.285714286)}) {3};
		\node (n4) at ({\R*cos(295.714285714)},{\R*sin(295.714285714)}) {4};
		\node (n5) at ({\R*cos(347.142857143)},{\R*sin(347.142857143)}) {5};
		\node (n6) at ({\R*cos(38.5714285714)},{\R*sin(38.5714285714)}) {6};
		
		\foreach \from/\to in {n0/n4,n0/n5,n0/n6,n1/n4,n1/n5,n1/n6,n2/n3,n2/n5,n2/n6,n3/n4,n3/n6}
		\draw (\from) -- (\to);
		\end{tikzpicture}
		\caption{$N=7$, $K=1$.}
	\end{subfigure}
	\hfill
	\begin{subfigure}[t]{.15\textwidth}
		\centering
		\def\R{2.8}
		\begin{tikzpicture}
		[scale=.3,auto=center,every node/.style={circle,fill=black!20,inner sep=1pt}]
		\node (n0) at ({\R*cos(90)},{\R*sin(90)}) {0};
		\node (n1) at ({\R*cos(141.428571429)},{\R*sin(141.428571429)}) {1};
		\node (n2) at ({\R*cos(192.857142857)},{\R*sin(192.857142857)}) {2};
		\node (n3) at ({\R*cos(244.285714286)},{\R*sin(244.285714286)}) {3};
		\node (n4) at ({\R*cos(295.714285714)},{\R*sin(295.714285714)}) {4};
		\node (n5) at ({\R*cos(347.142857143)},{\R*sin(347.142857143)}) {5};
		\node (n6) at ({\R*cos(38.5714285714)},{\R*sin(38.5714285714)}) {6};
		
		\foreach \from/\to in {n0/n2,n0/n3,n0/n4,n0/n5,n0/n6,n1/n2,n1/n3,n1/n4,n1/n5,n1/n6,n2/n4,n2/n5,n2/n6,n3/n4,n3/n5,n3/n6,n4/n6,n5/n6}
		\draw (\from) -- (\to);
		\end{tikzpicture}
		\caption{$N=7$, $K=2$.}
	\end{subfigure}
	\hfill
	\begin{subfigure}[t]{.15\textwidth}
		\centering
		\def\R{2.8}
		\begin{tikzpicture}
		[scale=.3,auto=center,every node/.style={circle,fill=black!20,inner sep=1pt}]
		\node (n0) at ({\R*cos(90)},{\R*sin(90)}) {0};
		\node (n1) at ({\R*cos(135)},{\R*sin(135)}) {1};
		\node (n2) at ({\R*cos(180)},{\R*sin(180)}) {2};
		\node (n3) at ({\R*cos(225)},{\R*sin(225)}) {3};
		\node (n4) at ({\R*cos(270)},{\R*sin(270)}) {4};
		\node (n5) at ({\R*cos(315)},{\R*sin(315)}) {5};
		\node (n6) at ({\R*cos(0)},{\R*sin(0)}) {6};
		\node (n7) at ({\R*cos(45)},{\R*sin(45)}) {7};
		
		\foreach \from/\to in {n0/n1,n0/n2,n0/n3,n0/n4,n0/n5,n0/n6,n0/n7,n1/n2,n1/n3,n1/n4,n1/n5,n1/n6,n1/n7,n2/n3,n2/n4,n2/n5,n2/n6,n2/n7,n3/n4,n3/n5,n3/n6,n3/n7,n4/n5,n4/n6,n4/n7,n5/n6,n5/n7,n6/n7}
		\draw (\from) -- (\to);
		\end{tikzpicture}
		\caption{$N=8$, $K=3$.}
	\end{subfigure}
	\caption{Minimum NCS graphs providing $K$-resilience computed by Algorithm~\ref{alg:quasi-minimum} under different $K$ and $N$ settings.}
	\label{fig:quasi}
\end{figure}
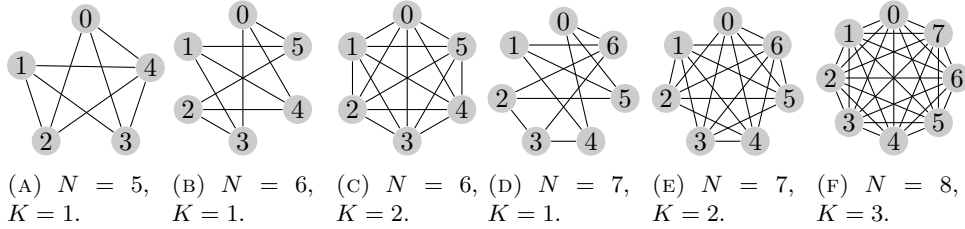

In this section, for any NCS graph $G=(V,E)$ that provides $K$-resilience, we derive a lower bound of $|E|$. We develop the following lemma that will be used to derive the bound.

\begin{lemma}
  A necessary condition for an NCS graph $G=(V,E)$ to give $K$-resilience is $\mathrm{mindeg}(G) \ge 2K+1$ where $\mathrm{mindeg}(G)$ denotes the minimum degree of all vertexes of $G$.
  \label{lemma:graph3}
\end{lemma}
\begin{proof}
  Theorem~\ref{thm:new1} shows that an NCS graph $G$ is $K$-resilient if and only if it is $(2K+1)$-edge-connected. From Whitney's theorem \cite{whitney1931theorem}, the minimum number of edges  whose deletion {\yellow results in disconnectivity of $G$ is {\yellow no greater} than the minimum degree $\mathrm{mindeg}(G)$. Thus, a necessary condition for an NCS graph $G=(V,E)$ to give $K$-resilience is $\mathrm{mindeg}(G) \ge 2K+1$.}
\end{proof}

\begin{theorem}
  For an NCS graph $G=(V,E)$ providing $K$-resilience, $|E| \ge \left\lceil \frac{N(2K+1)}{2} \right\rceil$, where $N=|V|$. In other words, a necessary condition for $G$ to be $K$-resilient is $|E| \ge \left\lceil \frac{N(2K+1)}{2} \right\rceil$.
  \label{pro:minum}
\end{theorem}

\begin{proof}
Let $C_1$ denote the clause that $G$ is $K$-resilient; let $C_2$ denote the clause that $\mathrm{mindeg}(G) \ge 2K+1$. Lemma~\ref{lemma:graph3} can be represented in logic as: $C_1 \Rightarrow C_2$.

  From Lemma~\ref{lemma:graph3}, to realize $K$-resilience, each node in $G$ should have a degree of at least $2K + 1$. We adopt a greedy algorithm to construct the graph $G^*$ with the minimum number of edges subject to the condition of $\mathrm{mindeg}(G^*) \ge 2K + 1$. The algorithm is as follows. Starting from no edges, each step of the algorithm adds an edge to connect two nodes that do not have an edge and the degree of each of them is no greater than any other nodes. The algorithm terminates once the condition $\mathrm{mindeg}(G^*) \ge 2K+1$ is satisfied. The resulting graph of this algorithm is as follows. First, when $N$ is an even number, the degree of every node is $2K+1$. The total number of edges is $\frac{N}{2} \cdot (2K+1)$. Second, when $N$ is an odd number, there are a total of $(N-1)$ nodes each having a degree of $(2K+1)$ and the remaining one node having a degree of $(2K+2)$. The total number of edges is $\left\lceil \frac{N(2K+1)}{2} \right\rceil$. In summary, the minimum number of edges to meet $\mathrm{mindeg}(G^*) \ge 2K + 1$ is $\left\lceil \frac{N(2K+1)}{2} \right\rceil$. Denoting by $C_3$ the clause of $|E| \ge \left\lceil \frac{N(2K+1)}{2} \right\rceil$, the above result can be represented in logic as $C_2 \Rightarrow C_3$.

Since $C_1 \Rightarrow C_2$ and $C_2 \Rightarrow C_3$, we have $C_1 \Rightarrow C_3$, i.e., $C_3$ is a necessary condition for $C_1$.
\end{proof}

The lower bound given by Theorem~\ref{pro:minum} can be used to understand the order of magnitude of the number of edges in a computed minimum NCS graph. Table~\ref{tab:bounds} shows the lower bound values under several settings of $N$ and $K$ as well as the numbers of the edges of the corresponding minimum NCS graphs shown in Fig.~\ref{fig:quasi}. We can see that the minimums are identical to the lower bound values, which {\yellow implies} that the {\yellow  theoretic lower bound is tight}. {\gray Thus, we can see that the relationship between the communication overhead (which is characterized by the number of edges) and the network size $N$ is roughly linear. Therefore, order-wise, the communication overhead for achieving $K$-resilience is acceptable. 
	
	In the traditional synchronization methods that do not provide any fault correction capbility, one slave node synchronizes with only one master node. Thus, the number of edges in the NCS graph of the traditional synchronization methods is $N-1$. From the result given by Theorem~\ref{pro:minum}, the additional communication overhead for $K$-resilience is at least $ \left\lceil \frac{N(2K+1)}{2} \right\rceil -(N-1) =  \left\lceil \frac{N(2K-1)}{2} \right\rceil +1$. For instance, when $K=1$ and $N=8$, the additional communication overhead is at least five p2p synchronization sessions.}

\begin{table}
  \caption{The lower bound of the number of edges in NCS graph providing $K$-resilience and the number of edges in the computed minimum NCS graphs shown in Fig.~\ref{fig:quasi}, as well as the upper bound of degree of resilience (DoR)  that is defined in \sect\ref{subsec:most-resilient}.}
  \label{tab:bounds}
  \centering
  \begin{tabular}{c c c c c}
    \hline\hline
    $N$ & $K$ & Lower bound & Number
                                                           of edges  & Upper bound \\
    & &  from Theorem~\ref{pro:minum} & in Fig.~\ref{fig:quasi} & of DoR \\
    \hline
    5 & 1 & 8 & 8 & $1/8$ \\
    6 & 1 & 9 & 9 & $1/9$ \\
    6 & 2 & 15 & 15 & $1/7.5$ \\
    7 & 1 & 11 & 11 & $1/11$ \\
    7 & 2 & 18 & 18 & $1/9$ \\
    8 & 3 & 28 & 28 & $1/9.\dot{3}$ \\
    \hline
  \end{tabular}
\end{table}

\section{Implication of Results}
\label{sec:implication}

This section discusses several important implications of the analytic results obtained in the previous sections.

\subsection{The Most Fault-Resilient Network}
\label{subsec:most-resilient}

If every p2p synchronization session has the same fault rate, the {\em degree of resilience} (DoR) defined as the ratio of the maximum number of correctable faults (i.e., $K$) and the number edges in an NCS graph (i.e., $|E|$) becomes a meaningful metric that characterizes the allowable percentage of faulty p2p synchronization sessions. We have the following corollary.

\begin{corollary}
  The 4-node network with complete NCS graph achieves the highest DoR.
\end{corollary}
\begin{proof}
We use the lower bound given by Theorem~\ref{pro:minum} to derive an upper bound of DoR when $N \ge 4$:
\begin{align}
  \mathrm{DoR} & \le \frac{K}{\left\lfloor \frac{N(2K+1)}{2} \right\rfloor - 1} \le \frac{K}{\frac{N(2K+1)}{2} - 2} = \frac{2K}{N(2K+1)-4} \nonumber \\
               & \le \frac{2K+1}{N(2K+1) - 4} = \frac{1}{N - \frac{4}{2K+1}} \le \frac{1}{N - \frac{4}{3}}, \label{eq:dor-upper}
\end{align}
where the last inequality follows from $K \ge 1$. Therefore, $\mathrm{DoR} = O \left(\frac{1}{N}\right)$, suggesting that larger networks will have lower degree of resilience when $N$ is large enough.

From Theorem~\ref{thm:1}, the DoR of the 4-node network is $1/6$. To ensure that the DoR upper bound given in Eq.~(\ref{eq:dor-upper}) is smaller than $1/6$ (i.e., $\frac{1}{N - 4/3} < 1/6$), we have $N \ge 8$. In other words, when $N \ge 8$, the network's DoR must be smaller than $1/6$. Now, we check the DoRs of the networks when $N \in [5,7]$. The last column of Table~\ref{tab:bounds} gives the upper bound of DoR that is the ratio of $K$ and the third column (i.e., the lower bound of the number of edges from Theorem~\ref{pro:minum}). From the results, we can see that when $N \in [5,7]$, the upper bound of DoR is smaller than $1/6$. Therefore, the 4-node network achieves the highest DoR of $1/6$.
\end{proof}

\subsection{Tiered Clock Synchronization for Fault Resilience}

The result in \sect\ref{subsec:most-resilient} suggests that, for a large-scale network, we can group the nodes into 4-node {\red synchronization groups}, forming the first tier of the clock synchronization. Each {\red synchronization group} with a complete NCS graph will use Algorithm~\ref{alg:error-correction2} to correct at most one fault. Every four central nodes from four tier-1 {\red synchronization groups} form a {\red synchronization group} in the second tier of the clock synchronization. Similarly, each tier-2 synchronization group will use Algorithm~\ref{alg:error-correction2} to correct at most one fault. More tiers are formed until all nodes in the network are connected. The NCS is executed from top to down in the tiered architecture.

We now use an example to illustrate. Suppose a network has 16 nodes. A two-tier clock synchronization with four 4-node tier-1 {\red synchronization groups} and one 4-node tier-2 {\red synchronization group} can be formed, as illustrated in Fig.~\ref{fig:16-node}. The tier-2 {\red synchronization group} executes NCS first. Then, each of the tier-1 {\red synchronization group} executes its own NCS. As such, all 16 nodes can be synchronized even if each {\red synchronization group} has a p2p synchronization faults (i.e., totally five faults). Note that the total number of edges in this two-tier network is 30. Alternatively, we can also use Algorithm~\ref{alg:quasi-minimum} to construct the minimum NCS graph for the 16-node network without the tiered architecture. From Theorem~\ref{pro:minum}, the minimum NCS graph that provides 5-resilience will have at least $\left\lceil \frac{16 \times (2 \times 5 + 1)}{2} \right\rceil = 88$ edges. Therefore, there is a trade-off between the above two solutions. In the minimum NCS graph without the tiered architecture, the five faults can occur on any five edges. However, the number of edges of the minimum NCS graph will be about three times of the tiered architecture shown in Fig.~\ref{fig:16-node}. On the other hand, while the tiered architecture uses less edges and thus incurs less communication cost, the five faults that the network can correct need to be distributed among the five synchronization groups.

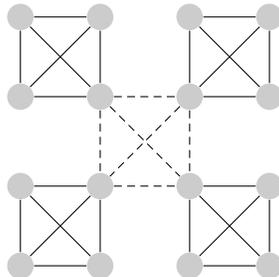
\begin{figure}
  \centering
  \def\R{2.5}
  \begin{tikzpicture}
    [scale=.3,auto=center,every node/.style={circle,fill=black!20,inner sep=1pt, minimum size=10pt}]
    \node (n0) at ({\R*cos(45)},{\R*sin(45)}) {};
    \node (n1) at ({\R*cos(135)},{\R*sin(135)}) {};
    \node (n2) at ({\R*cos(225)},{\R*sin(225)}) {};
    \node (n3) at ({\R*cos(315)},{\R*sin(315)}) {};

    \foreach \from/\to in {n0/n1,n0/n2,n0/n3,n1/n2,n1/n3,n2/n3}
    \draw (\from) -- (\to);

    \node (n4) at ({7.5+\R*cos(45)},{0+\R*sin(45)}) {};
    \node (n5) at ({7.5+\R*cos(135)},{0+\R*sin(135)}) {};
    \node (n6) at ({7.5+\R*cos(225)},{0+\R*sin(225)}) {};
    \node (n7) at ({7.5+\R*cos(315)},{0+\R*sin(315)}) {};

    \foreach \from/\to in {n4/n5,n4/n6,n4/n7,n5/n6,n5/n7,n6/n7}
    \draw (\from) -- (\to);

    \node (n8) at ({0+\R*cos(45)},{-7.5+\R*sin(45)}) {};
    \node (n9) at ({0+\R*cos(135)},{-7.5+\R*sin(135)}) {};
    \node (n10) at ({0+\R*cos(225)},{-7.5+\R*sin(225)}) {};
    \node (n11) at ({0+\R*cos(315)},{-7.5+\R*sin(315)}) {};

    \foreach \from/\to in {n8/n9,n8/n10,n8/n11,n9/n10,n9/n11,n10/n11}
    \draw (\from) -- (\to);

    \node (n12) at ({7.5+\R*cos(45)},{-7.5+\R*sin(45)}) {};
    \node (n13) at ({7.5+\R*cos(135)},{-7.5+\R*sin(135)}) {};
    \node (n14) at ({7.5+\R*cos(225)},{-7.5+\R*sin(225)}) {};
    \node (n15) at ({7.5+\R*cos(315)},{-7.5+\R*sin(315)}) {};

    \foreach \from/\to in {n12/n13,n12/n14,n12/n15,n13/n14,n13/n15,n14/n15}
    \draw (\from) -- (\to);

    \foreach \from/\to in {n3/n6,n3/n8,n3/n13,n6/n8,n6/n13,n8/n13}
    \draw[densely dashed] (\from) -- (\to);

  \end{tikzpicture}
  \caption{A two-tier 16-node clock synchronization architecture consisting of four 4-node tier-1 synchronization groups (solid lines) and one 4-node tier-2 synchronization groups (dashed lines).}
  \label{fig:16-node}
\end{figure}

\section{Conclusion}
\label{sec:conclude}

This paper studied the resilience of network clock synchronization based on practical p2p synchronization fault correction algorithms. Our analysis gave the following results:
\begin{enumerate}
  \item A closed-form tight bound of the maximum number of faults that can be corrected when every node pair in the network performs p2p synchronization, with respect to {\yellow the number of nodes $N$. The tight bound is $\left\lfloor \frac{N}{2} \right\rfloor - 1$}.
  \item An algorithm to compute the tight bound of the maximum number of faults that can be corrected when not every node pair performs p2p synchronization.
  \item A fast NCS algorithm {\yellow with a time complexity of $O(N^3)$} {\yellow that achieves the same fault correction capability as the original NCS algorithm that has an exponential time complexity.}
  \item An algorithm {\yellow that minimizes} the number of p2p synchronization sessions while ensuring that a specified number of faults can be corrected.
  \item A theoretic lower bound of the number of p2p synchronization sessions {\yellow needed to correct $K$ faults. The lower bound is $ \left\lceil \frac{N(2K+1)}{2} \right\rceil$}.
\end{enumerate}
Lastly, we showed that the 4-node network achieves the highest degree of resilience. Based on this, we discussed a tiered clock synchronization architecture that provides understood resilience and requires reduced p2p synchronization sessions. The results in this paper provide important understanding on the resilience of network clock synchronization against p2p synchronization faults and useful guidelines for the design of resilient clock synchronization systems.

\bibliographystyle{amsplain}

\providecommand{\bysame}{\leavevmode\hbox to3em{\hrulefill}\thinspace}
\providecommand{\MR}{\relax\ifhmode\unskip\space\fi MR }
\providecommand{\MRhref}[2]{%
  \href{http://www.ams.org/mathscinet-getitem?mr=#1}{#2}
}
\providecommand{\href}[2]{#2}

\end{document}